\documentclass[11pt]{article}
\usepackage{amssymb}
\usepackage{amsmath}
\usepackage{amsmath}
\usepackage{amsthm}
\usepackage{amsfonts}
\usepackage{graphicx}
\usepackage{enumerate}
\usepackage{bbm} 
\usepackage{dsfont}
\usepackage[authoryear,longnamesfirst]{natbib}
\usepackage{multirow}
\usepackage{multicol}
\usepackage{tikz}
\usepackage{comment}

\numberwithin{equation}{section}
\newtheorem{theorem}{Theorem}

\newtheorem{assumption}{Assumption}

\newtheorem{corollary}{Corollary}

\newtheorem{proposition}{Proposition}

\newtheorem{lemma}{Lemma}
\theoremstyle{definition}

\DeclareMathOperator{\var}{\text{Var}}
\DeclareMathOperator{\cov}{\text{Cov}}

\newcommand{\supp}{\rm{supp}}
\renewcommand{\hat}{\widehat}
\renewcommand{\tilde}{\widetilde}

\newcommand{\sumij}{\sum_{i=1}^N \sum_{j=1}^M}

\newcommand{\Real}{\mathbbm R}
\newcommand{\R}{\mathbbm R}

\newcommand{\calS}{\mathcal S}
\newcommand{\1}{\mathbbm 1}

\newcommand{\calG}{\mathcal G}

\begin{document}
	
\title{Algorithmic Subsampling under Multiway Clustering\thanks{\setlength{\baselineskip}{4.5mm}First arXiv version: February 28, 2021. We benefited from very useful comments by Sokbae (Simon) Lee, three anonymous referees, and participants in 2021 North American Summer Meeting, International Association for Applied Econometrics, 2021 Asian Meeting, 2021 China Meeting of the Econometric Society, $26^{\text{th}}$ International Panel Data Conference, 2021 Australasia Meeting of the Econometric Society, 2021 European Summer Meeting, and New York Camp Econometrics XVI. The usual disclaimer applies. We thank James M. Kilts Center, University of Chicago Booth School of Business for allowing us to use scanner data from the Dominicks Finer Foods (DFF) retail chain. H. Chiang is supported by the Office of the Vice Chancellor for Research and Graduate Education at the University of Wisconsin-Madison with funding from the Wisconsin Alumni Research Foundation.}}

\author{
	Harold D. Chiang\thanks{\setlength{\baselineskip}{4.5mm}Harold D. Chiang: hdchiang@wisc.edu.  Department of Economics, University of Wisconsin-Madison
	William H. Sewell Social Science Building, 1180 Observatory Drive, Madison, WI 53706, USA.} 
\qquad
	Jiatong Li\thanks{\setlength{\baselineskip}{4.5mm}Jiatong Li: jiatong.li@vanderbilt.edu. Department of Economics, Vanderbilt University, 
	VU Station B \#351819, 2301 Vanderbilt Place, Nashville, TN 37235-1819, USA} 
\qquad 
	Yuya Sasaki\thanks{\setlength{\baselineskip}{4.5mm}Yuya Sasaki: yuya.sasaki@vanderbilt.edu. Department of Economics, Vanderbilt University, 
	VU Station B \#35affiliated1819, 2301 Vanderbilt Place, Nashville, TN 37235-1819, USA}
}

\maketitle

\begin{abstract}\setlength{\baselineskip}{5.4mm}
This paper proposes a novel method of algorithmic subsampling (data sketching) for multiway cluster dependent data.
We establish a new uniform weak law of large numbers and a new central limit theorem for the multiway algorithmic subsample means.
We show that the algorithmic subsampling allows for robustness against potential degeneracy, and even non-Gaussian degeneracy, of the asymptotic distribution under multiway clustering at the cost of efficiency and power loss due to the algorithmic subsampling.
Simulation studies support this novel result, and demonstrate that inference with the algorithmic subsampling entails more accuracy than that without the algorithmic subsampling.
Applying these basic asymptotic theories, we derive the consistency and the asymptotic normality for the multiway algorithmic subsampling generalized method of moments estimator and for the multiway algorithmic subsampling M-estimator.
We illustrate an application to scanner data for analysis of differentiated products markets.
\end{abstract}

\newpage
\section{Introduction}\label{sec:introduction}

In the era of big data, it is not uncommon that data sets are so large that researchers may not need to use the whole sample for statistical inference to draw informative conclusions.
Furthermore, computational bottlenecks in time and/or memory may even prohibit econometric analyses with such large data sets.
The recent econometrics literature \citep*[e.g.,][]{LeeNg2020ARE,LeeNg2020sketching} suggests methods to deal with these circumstances that started to arise in today's data rich environments.
The \textit{algorithmic subsampling} or \textit{data sketching} explored by these authors paves the way for econometric and statistical analyses based on random subsampling of big data.

The existing study of the algorithmic subsampling focuses on i.i.d. cases, and it has been ``silent about how to deal with data that are dependent over time or across space'' \citep[][Section 8]{LeeNg2020ARE}.
On the other hand, some of big data may exhibit cross-sectional statistical dependence, such as multiway clustering.
For instance, common scanner data (leading examples of big data) are clustered in two ways by markets and products.
Common demand shocks within a market may induce statistical dependence among different products within that market.
Similarly, common supply shocks by a producer may induce statistical dependence among different markets within the product produced by that producer.
``A natural stochastic framework for the regression model with multiway clustered data is that of separately exchangeable random variables.'' \citep*{mackinnon2019wild}

In this paper, we propose a novel method of algorithmic subsampling for separately exchangeable random variables, which we will refer to as the \textit{multiway algorithmic subsampling}, and develop asymptotic statistical properties of this method.
We first establish basic theories for the multiway algorithmic subsample means, namely their uniform weak law of large numbers and central limit theorem, which differ from the standard ones in a meaningful way.
In particular, the form of the central limit theorem that is unique to the multiway algorithmic subsampling entails a practically useful property of robustness against potential degeneracy. 
Researchers do not know \textit{ex ante} how the data in use are affected by the cluster sampling.
In case that the cluster-specific shocks have no mean effect on the data, the standard multiway-cluster-robust asymptotic distribution \textit{without} the algorithmic subsampling would suffer from degeneracy, which can lead to either a Gaussian limiting distribution with a faster convergence rate or a non-Gaussian limiting distribution \citep[cf.][]{menzel2017bootstrap}, and invalidates the statistical inference based on standard multiway cluster-robust standard errors. 
On the other hand, we show that the multiway algorithmic subsampling allows for a non-degenerate asymptotic distribution regardless of whether the data are dependent or not. 
In other words, the algorithmic subsampling ensures a robustness against potential degeneracy, and even non-Gaussian degeneracy, of the asymptotic distribution, thereby allowing researchers to robustly enjoy valid statistical inference without knowing whether data are dependent or not.
This finding about the additional practical advantage of the algorithmic subsampling is novel in the literature to our knowledge.
We emphasize that these advantages of robustness come at the cost of efficiency and power loss due to the algorithmic subsampling. 

Once these basic asymptotic statistical theories are established, we apply them to common econometric frameworks.
Specifically, we propose a multiway algorithmic subsampling generalized method of moments (GMM) estimator, and derive asymptotic theories for it, including the consistency, asymptotic normality, and consistent variance estimation.
This multiway algorithmic subsampling GMM estimator also enjoys the aforementioned property of robustness against potential degeneracy.
Likewise, we also propose a multiway algorithmic subsampling M-estimator, and derive similar asymptotic theories for it, including the consistency, asymptotic normality, and consistent variance estimation.

{\bf Relation to the Literature:} 
This paper intersects with two branches of the literature, namely the algorithmic subsampling and multiway clustering.
In econometrics, the algorithmic subsampling and its properties are first studied by \citet*{LeeNg2020ARE,LeeNg2020sketching}.
This literature has focused on random (i.i.d.) sampling as emphasized earlier.
Robust variances under multiway clustering have been proposed by 
\citet*{cameron2012robust},
\citet*{thompson2011simple}, and
\citet*{cameron2014robust}.
Asymptotic statistical properties under multiway clustering have been rigorously investigated by
\citet*{davezies2019empirical},
\citet*{chiang2020inference},
\citet*{mackinnon2019wild}, 
\citet*{menzel2017bootstrap}, and
\citet*{chiang2019multiway}
under various contexts.
This literature has not considered the algorithmic subsampling. 
To our knowledge, this paper is the first to study the properties of algorithmic subsampling under multiway clustering, and therefore, is the first to propose the aforementioned advantage of the algorithmic subsampling for robustness against potential degeneracy, including non-Gaussian degeneracy, in the asymptotic distribution under multiway clustering.
We take advantage of the asymptotic distributional theory for incomplete one-sample U-statistics \citep*{Janson1984} to develop parts of our basic theoretical results.
The method of algorithmic subsampling is also closely related to the general scheme of resampling methods for clustered data, which has been studied for one or multiway clustering by \citet*{mackinnon2017wild}, \citet*{mackinnon2018wild}, \citet*{djogbenou2019asymptotic},
\citet*{davezies2019empirical}, \citet*{mackinnon2019wild}, \citet*{chiang2020inference}, and \citet*{menzel2017bootstrap}, to name but a few.

{\bf Organization:} 
The rest of this paper is organized as follows.
Section \ref{sec:multiway_algorithmic_subsampling} introduces the multiway algorithmic subsampling and presents its asymptotic statistical theories.
Sections \ref{sec:gmm} and \ref{sec:m} demonstrate applications to the GMM and M-estimation frameworks, respectively.
Section \ref{sec:simulation} presents Monte Carlo simulation studies.
Section \ref{sec:empirical} presents an empirical application to scanner data.
Section \ref{sec:conclusion} concludes.
All mathematical proofs and details are collected in the appendix.

\section{The Multiway Algorithmic Subsampling}\label{sec:multiway_algorithmic_subsampling}

Suppose that a researcher observes data $\{W_{ij} : 1 \le i \le N, 1 \le j \le M\}$, where $N$ and $M$ are the sample sizes in the first and second cluster dimensions.
For instance, $N$ and $M$ are the number of markets and the number of products, respectively, in scanner data.
The data may be two-way clustered, in the sense that we allow for arbitrary statistical dependence of $W_{ij}$ across $j \in \{1,...,M\}$ within each market $i$ (due to a common demand shock) and also allow for arbitrary statistical dependence of $W_{ij}$ across $i \in \{1,...,N\}$ within each product $j$ (due to a common supply shock).
A formal assumption of this sampling process will be stated as Assumption \ref{a:sampling} ahead.
Using such two-way cluster sampled data, we are interested in (uniformly) consistent estimation of and statistical inference about $E[f(W_{ij})]$ based on standard econometric techniques.
Since scanner data are very big, however, computational bottlenecks in time and/or memory may limit or even prohibit implementation of standard econometric analysis. \citet*{LeeNg2020ARE,LeeNg2020sketching} therefore suggest the algorithmic subsampling of big data to alleviate computational burdens.

Adapting the ideas of \citet*{LeeNg2020ARE,LeeNg2020sketching} to our framework of two-way clustered data, we propose the following multiway algorithmic subsampling procedure.
(We remark that the algorithmic subsampling is different from the subsampling as a resampling method.) 
Let $p_{NM}$ denote the probability of subsample selection that may depend on the current sample size $(N,M)$.
Generate i.i.d Bernoulli$\left(p_{NM}\right)$ random variables $\{Z_{ij}: 1 \le i \le N, 1 \le j \le M\}$ independently from data.
Let $\hat L = \sum_{i=1}^N \sum_{j=1}^M Z_{ij}$ denote the number of non-zero elements.
Note that $\hat L$ follows Binomial$\left(NM,p_{NM}\right)$, and thus $L \equiv E[\hat L]=NMp_{NM}$ in particular.
In fact, this formulation of the algorithmic subsampling is called the Bernoulli subsampling, and is one of the alternative approaches to subsampling proposed by \citet*{LeeNg2020ARE}.
We focus on this Bernoulli subsampling in this paper for simplicity as well as its desired property of the aforementioned robustness against degeneracy.
That said, we remark that it is also feasible to use alternative subsampling methods (namely, the uniform subsampling with and without replacement) proposed by \citet*{LeeNg2020ARE}.

We use
$
\hat L^{-1}\sum_{i=1}^{N}\sum_{j=1}^{M}Z_{ij}f\left(W_{ij}\right)
$
to estimate and make inference about $E[f(W_{ij})]$.
To this end, we first develop the uniform weak law of large numbers and the central limit theorem under this setting of the multiway algorithmic subsampling in Sections \ref{sec:wlln} and \ref{sec:clt}, respectively.
We then apply these basic theories in turn to establish the consistency and the asymptotic normality for the generalized method of moments and M-estimation in Sections \ref{sec:gmm} and \ref{sec:m}, respectively. Hereafter for conciseness of notations, $p_{NM}$ will be abbreviated as $p$. We let $[k]$ denote the set $\{1,...,k\}$ for any $k \in \mathbb{N}$, and let $[k]^c=\mathbb{N}\backslash[k]$ for any $k\in \mathbb{N}.$
We use the short-hand notation $\underline{C}=\min\left\{N,M\right\}$. Throughout the paper, the asymptotics is understood as $\underline{C}\to \infty$. For a vector $v\in \Real^k$, let $\|v\|$ denote the Euclidean norm of $v$. 

\subsection{The Uniform Weak Law of Large Numbers}\label{sec:wlln}

We first formally state the assumption of two-way cluster sampling.

\begin{assumption}[Sampling]\label{a:sampling} 
(i) $\left(W_{ij}\right)_{(i,j)\in \mathbb{N}^2}$ is an infinite sequence of separately exchangeable d-dimensional random vectors. That is, for any permutations $\pi_1$ and $\pi_2$ of $\mathbb{N}$, we have $\left(W_{ij}\right)_{(i,j)\in \mathbb{N}^2} \overset{\text{d}}{=} \left(W_{\pi_1\left(i\right)\pi_2\left(j\right)}\right)_{\left(i,j\right)\in \mathbb{N}^2}.$
(ii) $\left(W_{ij}\right)_{\left(i,j\right)\in \mathbb{N}^2}$ is dissociated. That is, for any $\left(c_1,c_2\right)\in \mathbb{N}^2$, $\left(W_{ij}\right)_{i\in [c_1],j\in [c_2]}$ is independent of $\left(W_{ij}\right)_{i\in [c_1]^c,j\in [c_2]^c}.$
\end{assumption}

\noindent
Part (i) requires a form of the identical distribution condition in separate permutations of the $i$ index and the $j$ index.
Although we relax the independent sampling, we maintain a form of the identical distribution as such.
Part (ii) requires that sets of observations are independent if they do not share the same $i$ index or the same $j$ index, i.e., $\left(W_{ij}: i \in \{1,...,c_1\},j \in \{1,...,c_2\}\right)$ and $\left(W_{ij}: i \in \{c_1+1,...\},j \in \{c_2+1,...\}\right)$ are assumed to be independent for any $(c_1,c_2)\in \mathbb{N}^2$. However, any observations are sharing either the same $i$ index or the same $j$ index, then they are allowed to be arbitrarily dependent.
For example, in the scanner data, two observations in the same market $i$ may be dependent due to a common demand shock, and likewise two observations in the same product $j$ may also be dependent due to a common supply shock.

Assumption \ref{a:sampling} consists of a sufficient condition for what we actually need. These conditions can be relaxed to the assumption that the data $(W_{ij})_{(i,j)\in\mathbb{N}^2}$ are generated via the process $W_{ij} = f(\alpha_i,\beta_j,\varepsilon_{ij})$ for some Borel-measurable function $f$, where $(\alpha_i)_{i \in \mathbb{N}}$, $(\beta_j)_{j \in \mathbb{N}}$, and $(\varepsilon_{ij})_{(i,j) \in \mathbb{N}^2}$ are mutually independent, and each of $(\alpha_i)_{i \in \mathbb{N}}$, $(\beta_j)_{j \in \mathbb{N}}$, and $(\varepsilon_{ij})_{(i,j) \in \mathbb{N}^2}$ is i.i.d. This data generating process, or so-called the Aldous-Hoover-Kallenberg representation, is implied by Assumption \ref{a:sampling}. We can interpret $\alpha_i$ and $\beta_j$ as $i$- and $j$-specific effects, respectively, while $\varepsilon_{ij}$ is an idiosyncratic effect. This representation is also consistent with the data generating processes considered in the simulation studies.

For convenience of stating the next assumption, we introduce additional notations and definitions.
Let $(T,d)$ be pseudometric space\footnote{That is, $d(x,y)=0$ does not imply $x=y$}. For $\varepsilon>0$, an $\varepsilon$-net of $T$ is a subset $T_\varepsilon$ of $T$ such that for every $t\in T$ there exists $t_\varepsilon\in T_\varepsilon$ with $d(t,t_\varepsilon)\le \varepsilon$. 
We define the $\varepsilon$-covering number $N(T,d,\varepsilon)$ of $T$ by
\begin{align*}
N(T,d,\varepsilon)=\inf \{\text{Card}(T_\varepsilon):\, T_\varepsilon \text{ is an $\varepsilon$-net of $T$}\}.
\end{align*}
For any probability measure $Q$ on a measurable space $(S,\calS)$ and any $q\ge 1$, define $\|f\|_{Q,q}=\left\{\int |f|^q dQ\right\}^{1/q}$ and let $L^q(S) = \{f:S\to \Real : \|f\|_{Q,q} < \infty\}$. A function $G:S\to\R$ is an envelope of a class of functions $\calG\ni g$, $g:S\to\R$, if $\sup_{g\in\calG}|g(s)|\le G(s)$ for all $s\in S$. 
With these notations and definitions, we state the following assumption regarding the function class where $f$ resides.

\begin{assumption}[Function Class]\label{a:moment}
The function class $\mathcal{F}$ satisfies (i) $E\left[f\left(W_{ij}\right)\right]=0$ for all $f\in \mathcal{F}$.
(ii) $\mathcal{F}$ admits an envelope $F$ satisfying $E\left[F(W_{ij})\right]<\infty$ with 
$\sup_{Q} N(\mathcal{F},\left\|\cdot\right\|_{Q,2},\epsilon \left\|F\right\|_{Q,2})<\infty $ for all $\epsilon >0$, where $Q$ is any finite discrete measure.
(iii) $\mathcal{F}$ is pointwise measurable.\footnote{For its definition, see \citet[pp. 110]{van1996weak} for instance.} 
\end{assumption}

\noindent 
Part (i) is a location normalization (centering) and is therefore without loss of generality.
Although this part will not be needed in the short run (Lemma \ref{lemma:consistency}), we state it here as this Assumption \ref{a:moment} collects requirements about the function space where $f$ resides.
Part (ii) is a regularity condition imposed to establish a uniform weak law of large numbers. 
Part (iii) is a technical requirement that is used to avoid measurability issues.
At this moment, we are stating these high-level assumptions for the sake of generality, but we will provide the standard lower-level primitive conditions in the contexts of the application to the generalized method of moments presented in Section \ref{sec:gmm} and the application to the M-estimation presented in Section \ref{sec:m}.

Under these assumptions, we can establish the uniform weak law of large numbers for multiway algorithmic subsample means as formally stated in the lemma below.

\begin{lemma}[Uniform Weak Law of Large Numbers]\label{lemma:consistency}
Suppose that Assumption \ref{a:sampling} holds and that $\mathcal{F}$ satisfies Assumption \ref{a:moment} (ii)--(iii).
Then, for any $f \in \mathcal{F}$, we have
$$ \sup_{f\in \mathcal{F}}\left|\frac{1}{\hat L}\sum_{i=1}^{N}\sum_{j=1}^{M}Z_{ij}f\left(W_{ij}\right)- E\left[f(W_{11})\right]\right|\overset{P}{\rightarrow} 0.$$
\end{lemma}

\noindent
This result is not very surprising, but we state above as Lemma \ref{lemma:consistency} and prove it (in Appendix \ref{sec:lemma:consistency}) for the following two reasons.
First, this is nonetheless the first time it is stated and proved in the literature to the best of our knowledge.
Second, more importantly, this lemma serves as a useful auxiliary device for other results to be presented ahead that are practically more relevant.

\subsection{The Central Limit Theorem}\label{sec:clt}

To establish the central limit theorem under the multiway algorithmic subsampling, we augment our assumptions with the following additional condition.
Recall the short-hand notation $\underline{C}=\min\left\{N,M\right\}$.
\begin{assumption}\label{a:scalar}
There exists a constant $\Lambda \geq 0$ such that $(\underline{C}/NM)((1-p_{NM})/p_{NM}) \rightarrow \Lambda$. 
\end{assumption} 

\noindent
 It entails that there exist constants $\lambda_1 \geq 0$ and $\lambda_2\geq 0$ such that $\underline{C}/N\rightarrow \lambda_1$, $\underline{C}/M\rightarrow \lambda_2.$
 
To facilitate the subsequent discussions, we introduce a notion of degenerate asymptotic distribution.
For any scalar-valued sequence of random variables $\left(X_{ij}\right)_{(i,j)\in \mathbb{N}^2}$ satisfying Assumption \ref{a:sampling}, we say the asymptotic distribution is degenerate if $\var\left( (\sqrt{\underline C} /NM)\sumij X_{ij}\right)\to 0$ as $\underline C\to\infty$.
The following theorem establishes the central limit theorem under the multiway algorithmic subsampling.

\begin{theorem}[Central Limit Theorem]\label{theorem:clt}
Suppose that Assumptions \ref{a:sampling}, \ref{a:moment} (i), (iii) and \ref{a:scalar} hold, and that the class $\mathcal{F}=\left\{f_1,...,f_k\right\}$ is finite, independent of sample size, and admits an envelope $F$ satisfying $E\left[F(W_{ij})^2\right]<\infty$. Let $f=\left(f_1,...,f_k\right)^T$. 
Then,
$$
\sqrt{\underline{C}}\frac{1}{\hat L}\sum_{i=1}^{N}\sum_{j=1}^{M}Z_{ij}f\left(W_{ij}\right)\overset{d}{\rightarrow}N\left(0,\Gamma\right),
$$
where the variance is given by $\Gamma=\Gamma _A+\Lambda \Gamma_B,$ 
$\Gamma_A=\lambda_1 E\left[f\left(W_{11}\right)f^T\left(W_{12}\right)\right]+\lambda_2 E\left[f\left(W_{11}\right)f^T\left(W_{21}\right)\right]$
and   
$\Gamma_B=E\left[f\left(W_{11}\right)f^T\left(W_{11}\right)\right].$
\end{theorem}

A proof is provided in Appendix \ref{sec:theorem:clt}.
This central limit theorem entails a novel and useful feature of the algorithmic subsampling for multiway clustered data in practice.
Notably, the asymptotic variance $\Gamma$ consists of a sum of two components, $\Gamma_A$ and $\Lambda \Gamma_B$.
This is in contrast with the algorithmic subsampling for independent data, where the asymptotic variance consists of only one term.
The first part, $\Gamma_A$, in fact coincides with the asymptotic variance that we would get without algorithmic subsampling.
Specifically, $\Gamma = \Gamma_A$ if $p = 1$.
More generally, if $p$ is a constant, as the sample sizes $(N,M)$ increase, then $\Lambda = 0$ so that $\Gamma = \Gamma_A$.
On the other hand, if $p$ is chosen so that $\Lambda = \lim_{N,M \rightarrow \infty} (\underline{C}/NM)((1-p)/p) > 0$, then the second part, $\Lambda \Gamma_B$, is also present.
Furthermore, $\Gamma_B$ is nonzero whenever the distribution of $f(W_{ij})$ is non-degenerate, and this feature provides a practically useful property of the robustness in inference against possible events of no cross sectional dependence.

In practice, a researcher may not \textit{ex ante} know whether data exhibit cross sectional dependence ($E[f(W_{11})f^T(W_{12})] \neq 0$ or $E[f(W_{11})f^T(W_{21})] \neq 0$) or not.
If a researcher knew the true dependence structure, he or she could set the correct cluster dimension to conduct valid inference.
However, this premise is implausible.
In case where there is no cross sectional dependence, then $\Gamma_A=0$ and the statistical inference based on the asymptotic normality \textit{without} the algorithmic subsampling would suffer from the degeneracy problem.
Because of the algorithmic subsampling, however, we can robustly safeguard against such degenerate asymptotic distributions without requiring a prior knowledge of the researcher about the presence/absence of cross sectional dependence in data.
This result is novel in the literature, and also uncovers an additional useful property of the algorithmic subsampling in practice.\footnote{Indeed, the method of inference by \citet*{mackinnon2019wild} as well as \citet*{cameron2012robust} adapts to specific classes of degenerate asymptotic distributions. 
However, these restrict to the cases of Gaussian degeneracy, where the convergence rate is $\sqrt{NM}$ yet the asymptotic distribution is still Gaussian. 
On the other hand, these existing methods of inference by \citet*{cameron2012robust} and \citet*{mackinnon2019wild} do not adapt to the class of non-Gaussian degenerate asymptotic distributions.}
Simulation studies presented in Section \ref{sec:simulation} support this practically relevant property of the multiway algorithmic subsampling.

Intuitively, the algorithmic subsampling with smaller $p$ makes it less likely that multiple observations from the same row $i$ or same column $j$ are selected.
Thus, it results in placing relatively more weights on the variance $E[f(W_{ij}) f(W_{ij})]$ than on the covariances, $E[f(W_{ij}) f(W_{ij'})]$ and $E[f(W_{ij}) f(W_{i'j})]$, in the asymptotic distribution.
Hence, the part $E[f(W_{ij}) f(W_{ij})]$ of the asymptotic variance becomes dominant in the case of degenerate covariances, and this feature of the algorithmic subsampling prevents the degeneracy problem.

We can apply these theoretical results to a number of common frameworks of econometric analysis.
Two of the most frequently used classes of econometric methods are the generalized method of moments (GMM) and the M-estimation.
Therefore, we will demonstrate applications of these basic theories of the uniform weak law of large numbers (Lemma \ref{lemma:consistency}) and the central limit theorem (Theorem \ref{theorem:clt}) to establish the consistency and the asymptotic normality of the GMM and M-estimators under the multiway algorithmic subsampling in Sections \ref{sec:gmm} and \ref{sec:m}. 

We conclude this section with a remark on alternative subsampling methods.
As mentioned earlier, our method is based on the Bernoulli subsampling, and is one of the alternative approaches to subsampling proposed by \citet*{LeeNg2020ARE}.
Besides the Bernoulli subsampling on which we focus in this paper, they propose the uniform subsampling with replacement, the uniform subsampling without replacement, and the leverage score subsampling.
Among these alternative methods, it is also feasible to use the uniform subsampling with replacement and the uniform subsampling without replacement.
Similar asymptotic properties will follow through similar lines of the argument following \citet{Janson1984} to those in the proof of Theorem \ref{theorem:clt}.
See Appendix \ref{sec:multiple} for details.

\subsection{Application to the Ordinary Least Squares}\label{sec:ols}
This section demonstrates an application of the basic theories to the ordinary least squares (OLS) estimator. 
Consider the linear regression model
\begin{equation}
Y_{ij} = X_{ij}^T\beta + u_{ij}
\qquad 1 \le i \le N, 1 \le j \le M,    
\end{equation}
where $Y_{ij}$ is a response variable, $X_{ij}$ is a vector of $d$ covariates and $u_{ij}$ is an error satisfying $E\left[u_{ij}|X_{ij}\right]=0.$ 
Let $W_{ij}=\left(Y_{ij},X_{ij}^T\right)^T$, and we apply the proposed multiway algorithmic subsampling to $W_{ij}$.
The parameter of interest is the vector of linear projection coefficients
\begin{equation}
    \beta = E\left[X_{11}X_{11}^T\right]^{-1}E\left[X_{11}Y_{11}\right], \label{beta}
\end{equation}
and the multiway algorithmic subsampling OLS estimator is
\begin{equation}
    \hat \beta=\left(\frac{1}{\hat L}\sum_{i=1}^{N}\sum_{j=1}^{M}Z_{ij}X_{ij}X_{ij}^T\right)^{-1}\left(\frac{1}{\hat L}\sum_{i=1}^{N}\sum_{j=1}^{M}Z_{ij}X_{ij}Y_{ij}\right). \label{hatbeta}
\end{equation}
Applications of Lemma \ref{lemma:consistency} and Theorem \ref{theorem:clt} yield the following limit distribution property about $\widehat\beta$.

\begin{corollary}\label{pro:OLS}
Suppose that Assumption \ref{a:sampling} holds for $W_{ij}=\left(Y_{ij},X_{ij}^T\right)^T.$ Assume $E\left[|Y_{11}|^4\right]<\infty,$ $E\left[\left|\left|X_{11}\right|\right|^4\right]<\infty,$ and that $E\left[X_{11}X_{11}^T\right]$ non-singular. 
For $\beta$ and $\hat \beta$ defined in \eqref{beta} and \eqref{hatbeta}, we have \begin{equation*}
    \sqrt{\underline C}\left(\hat \beta-\beta\right)\overset{d}{\rightarrow}N\left(0,V\right),
\end{equation*}
where 
$V=J^{-1}\Gamma_{OLS} J^{-1}$, 
$J=E\left[X_{11}X_{11}^T\right]$,
$\Gamma_{OLS}=\Gamma_{OLS,1}+\Lambda \Gamma_{OLS,2}$, 
$\Gamma_{OLS,1}=\lambda_1 E\left[X_{11}u_{11}\left(X_{12}u_{12}\right)^T\right]+\lambda_2 E\left[X_{11}u_{11}\left(X_{21}u_{21}\right)^T\right]$,
and   
$\Gamma_{OLS,2}=E\left[X_{11}u_{11}\left(X_{11}u_{11}\right)^T\right]$.
\end{corollary}
\noindent
See Appendix \ref{sec:pro:OLS} for a proof of this corollary.

\section{Application to the Generalized Method of Moments (GMM)}\label{sec:gmm}
In this section, we apply the basic methods and theories presented in Section \ref{sec:multiway_algorithmic_subsampling} to the multiway algorithmic subsampling generalized method of moments (GMM).
Suppose that an economic model implies moment restrictions $E\left[g\left(W_{ij},\theta^0\right)\right]=0$ for a true parameter vector $\theta^0=(\theta_1^0,...,\theta_k^0)^T \in \Theta$, $\Theta \subset \mathbb{R}^k,$ where $g =\left(g_1 ,...,g_m \right)^T$, and $m\geq k$. 
With a Bernoulli sample $\{Z_{ij} : 1 \le i \le N, 1 \le j \le M\}$, the algorithmic subsample moment evaluated at $\theta=(\theta_1,...,\theta_k)^T \in \Theta$ is given by $\hat g_{NM}\left(\theta\right)=\hat L^{-1}\sum_{i=1}^{N}\sum_{j=1}^{M}Z_{ij}g\left(W_{ij},\theta\right)$.  
Let $\hat V$ be a positive semi-definite random matrix, which may depend on $\theta$. We define the multiway algorithmic subsampling GMM estimator $\hat\theta$ as the solution to 
$$
\max_{\theta \in \Theta}\hat Q_{NM}\left(\theta\right),
$$ 
where  $\hat Q_{NM}\left(\theta\right)=-\hat g_{NM}\left( \theta\right)^T\hat V\hat g_{NM}\left(\theta\right).$ The true parameter vector $\theta^0\in \Theta$ is assumed to uniquely solve the population problem $\max_{\theta \in \Theta}-E[g(W_{ij},\theta)]^TVE[g(W_{ij},\theta)],$ where $V$ is positive semi-definite and $\hat V \overset{P}{\rightarrow}V$.

\subsection{Consistency and Asymptotic Normality}

To establish the consistency and asymptotic normality for the multiway algorithmic subsampling GMM estimator $\widehat\theta$, we make the following assumption. For concisely stating the following assumption, we introduce one additional definition regarding Lipschitz continuity. A function $g: \mathbb{R}^k \to\R,$ is Lipschitz with a universal Lipschitz constant, if there exists a positive constant $M$ such that
$\left|g\left(w,\theta\right)-g\left(w,\theta'\right)\right|\leq M \left\|\theta-\theta'\right\|$
for all $w\in\supp(W_{ij})$.

\begin{assumption}\label{a:GMM}
~\\
(i) $V$ is positive semi-definite, and $VE[g(W_{ij},\theta)]=0$ only if $\theta=\theta^0.$

\noindent (ii) $\theta^0 \in \textrm{int} \left(\Theta\right)$, where $\Theta$ is a compact subset of $\mathbb{R}^k$.

\noindent (iii) (a) $\theta \mapsto g_{r}(w, \theta)$ is Lipschitz with a universal Lipschitz constant. 

\noindent\hspace{0.67cm}
(b) Each coordinate of $\theta \mapsto \nabla_{\theta}g_r(w,\theta)$ is Lipschitz with a universal Lipschitz constant. 

\noindent (iv) $E\left[\sup_{\theta\in \Theta}\left\|g\left(W_{ij},
\theta\right)\right\|\right]<\infty.$

\noindent (v) $G^TVG$ is nonsingular where $G=E\left[\nabla_{\theta}g\left(W_{ij},\theta^0\right)\right].$

\noindent (vi) $E\left[\sup_{\theta \in\Theta}\left\|\nabla_\theta g\left(W_{ij},\theta\right)\right\|\right]<\infty.$

\noindent (vii) $g_{\sup}(\cdot)=\max_{r\in \{1,...,m\}}|g_r\left(\cdot,\theta\right)|$ satisfies $E[g_{\sup}(W_{ij})^2]<\infty$.
\end{assumption}

\noindent Assumption \ref{a:GMM} is analogous to the conditions required for Theorem 2.6 and Theorem 3.4 in \citet*{Newey1994}, which state the consistency and asymptotic normality, respectively, of the GMM estimator under the conventional random sampling. 

We first state the consistency of the multiway algorithmic subsampling GMM estimator.

\begin{lemma}[Consistency of the Multiway Algorithmic Subsampling GMM Estimator]\label{theorem:consistency_gmm}
If Assumptions \ref{a:sampling} and \ref{a:GMM} (i), (ii), (iii), (iv) hold, and that $\hat V\overset{P}{\rightarrow} V$,
then $\hat \theta \overset{P}{\rightarrow} \theta^0.$
\end{lemma}

\noindent
A proof is provided in Appendix \ref{sec:theorem:consistency_gmm}.
It follows from combining the arguments in the proofs of \citet[][Theorem 2.1]{Newey1994} with our uniform weak law of large numbers for the multiway algorithmic subsampling (Lemma \ref{lemma:consistency}) presented in Section \ref{sec:wlln}.

We next state the asymptotic normality of the multiway algorithmic subsampling GMM estimator.

\begin{theorem}[Asymptotic Normality of the Multiway Algorithmic Subsampling GMM Estimator]\label{theorem:asymptotic_normality_gmm}
If Assumptions \ref{a:sampling}, \ref{a:scalar}, and \ref{a:GMM} hold, and that $\hat V\overset{P}{\rightarrow} V$,
then $$\sqrt{\underline{C}}\left(\hat \theta-\theta^0\right)\overset{d}{\rightarrow} N\left(0,\left(G^TVG\right)^{-1}G^TV\Omega VG\left(G^TVG\right)^{-1}\right),$$
where $\, G=E\left[\nabla_{\theta}g\left(W_{11},\theta^0\right)\right]$ and $\, \Omega=\Gamma _1+\Lambda \Gamma_2$, with $\, \Gamma _1= \lambda_1E\left[g\left(W_{11},\theta^0\right)g^T\left(W_{12},\theta^0\right)\right]+\lambda_2E\left[g\left(W_{11},\theta^0\right)g^T\left(W_{21},\theta^0\right)\right]$ and $ \Gamma_2=E\left[g\left(W_{11},\theta^0\right)g^T\left(W_{11},\theta^0\right)\right]$.
\end{theorem}

A proof is provided in Appendix \ref{sec:theorem:asymptotic_normality_gmm}.
It follows from combining the arguments in the proofs of \citet[][Theorem 3.4]{Newey1994} with our central limit theorem for the multiway algorithmic subsampling (Theorem \ref{theorem:clt}) presented in Section \ref{sec:clt}.

\subsection{Algorithmic Subsampling Variance Estimation}

\noindent The components, $G$ and $\Omega$ in Theorem \ref{theorem:asymptotic_normality_gmm}, of the asymptotic variance of the multiway algorithmic subsampling GMM estimator can be estimated by 
$$
\widetilde G=\frac{1}{\hat L}\sum_{i=1}^{N}\sum_{j=1}^{M}Z_{ij}\nabla_{\theta}g\left(W_{ij},\hat \theta\right)
$$ 
and 
$$
\widetilde \Omega=\widetilde \Gamma_1+\Lambda \widetilde \Gamma_2,
$$
respectively,
where 
$$
\widetilde \Gamma_1=\frac{\underline{C}}{\widehat L^2}\sum_{i=1}^{N}\sum_{1\leq j,j'\leq M}Z_{ij}Z_{ij'}g\left(W_{ij},\hat \theta\right)g^T\left(W_{ij'},\hat \theta\right)+ \frac{\underline{C}}{\widehat L^2}\sum_{1\leq i, i'\leq N}\sum_{j=1}^MZ_{ij}Z_{i'j}g\left(W_{ij},\hat \theta\right)g^T\left(W_{i'j},\hat \theta\right)
$$ 
and 
$$
\widetilde \Gamma_2= \frac{1}{\hat L}\sum_{i=1}^{N}\sum_{j=1}^{M}Z_{ij}g\left(W_{ij},\hat \theta\right)g^T\left(W_{ij},\hat \theta\right).
$$ 

We propose to estimate the asymptotic variance $\left(G^TVG\right)^{-1}G^TV\Omega VG\left(G^TVG\right)^{-1}$ by the sample counterpart $\left(\tilde G^T\hat V\tilde G\right)^{-1}\tilde G^T\hat V\tilde \Omega \hat V\tilde G\left(\tilde G^T\hat V\tilde G\right)^{-1}.$ 
To guarantee that this algorithmic subsampling variance estimator works asymptotically, we make the following assumption in addition.

\begin{assumption}\label{a:GMM:variance}
~\\
(i) $\theta \mapsto E\left[\nabla_{\theta}g\left(W_{ij},\theta\right)\right]$ is continuous at $\theta^0.$

\noindent (ii) $\theta \mapsto \lambda_1E\left[g\left(W_{ij},\theta\right)g^T\left(W_{ij},\theta\right)\right]+\lambda_2E\left[g\left(W_{ij},\theta\right)g^T\left(W_{ij},\theta\right)\right]$ is continuous at $\theta^0.$

\noindent (iii) $\theta \mapsto E\left[g\left(W_{ij},\theta\right)g^T\left(W_{ij},\theta\right)\right]$ is continuous at $\theta^0.$
\end{assumption}

\noindent
With this additional assumption, $\left(\tilde G^T\hat V\tilde G\right)^{-1}\tilde G^T\hat V\tilde \Omega \hat V\tilde G\left(\tilde G^T\hat V\tilde G\right)^{-1}$ is consistent for the asymptotic variance $\left(G^TVG\right)^{-1}G^TV\Omega VG\left(G^TVG\right)^{-1}$, as formally stated in the following theorem. 

\begin{theorem}[Consistent Asymptotic Variance Estimation of the Multiway Algorithmic Subsampling GMM Estimator]\label{theorem:consistency_gmm_variance}
If Assumptions \ref{a:sampling}, \ref{a:scalar}, \ref{a:GMM} and \ref{a:GMM:variance} hold and that $\hat V\overset{P}{\rightarrow}  V$, 
then 
$$
\left(\tilde G^T\hat V\tilde G\right)^{-1}\tilde G^T\hat V\tilde \Omega \hat V\tilde G\left(\tilde G^T\hat V\tilde G\right)^{-1} \overset{P}{\rightarrow}\left(G^TVG\right)^{-1}G^TV\Omega VG\left(G^TVG\right)^{-1}.
$$
\end{theorem}

\noindent
A proof is provided in Appendix \ref{sec:theorem:consistency_gmm_variance}. 
It follows by combining Lemma \ref{lemma:consistency} and similar lines of arguments to those in the proofs of Lemma \ref{theorem:consistency_gmm} and Theorem \ref{theorem:asymptotic_normality_gmm}.

\section{Application to the M-Estimation}\label{sec:m}
In this section, we apply the basic methods and theories presented in Section \ref{sec:multiway_algorithmic_subsampling} to the multiway algorithmic subsampling M-estimation.
Let $\Theta \subset \mathbb{R}^k$ be a parameter space and define the class $\mathcal{Q}=\left\{q\left(\cdot, \theta\right): \theta \in \Theta\right\}$ of functions $q(\cdot,\theta)$ indexed by $\theta$. 
With a Bernoulli sample $\{Z_{ij}, 1 \leq i \leq N,  1 \leq j \leq M\}$, we define the multiway algorithmic subsampling M-estimator $\hat \theta$ as the solution to 
$$
\max_{\theta \in \Theta}-\frac{1}{\hat L}\sum_{i=1}^{N}\sum_{j=1}^{M}Z_{ij}q\left(W_{ij},\theta\right).
$$
The true parameter vector $\theta^0=(\theta_1^0,...,\theta_k^0)^T \in \Theta$ is assumed to uniquely solve the population maximization problem $\max_{\theta \in \Theta}-E\left[q\left(W_{ij},\theta\right)\right]$, in the sense that $E\left[q\left(W_{ij},\theta^0\right)\right]<E\left[q\left(W_{ij},\theta\right)\right]$ holds for all $\theta=(\theta_1,...,\theta_k)^T \in \Theta$ and $\theta \neq \theta^0$. 
For each $\theta \in \Theta$, let $-\hat L^{-1}\sum_{i=1}^{N}\sum_{j=1}^{M}Z_{ij}q\left(W_{ij}, \theta\right)$ and $-E\left[q\left(W_{ij},\theta\right)\right]$ be denoted by $\hat Q_{NM}\left(\theta\right)$ and $Q_0\left(\theta\right)$, respectively, for conciseness.

\subsection{Consistency and Asymptotic Normality}

To establish the consistency and asymptotic normality for the multiway algorithmic subsampling M-estimator $\widehat\theta$, we make the following assumption.

\begin{assumption}\label{a:M}
~\\
(i) $\theta^0 \in \textrm{int} \left(\Theta\right)$ where $\Theta$ is a compact subset of $\mathbb{R}^k$, and $E[q(W_{ij}, \theta^0)]<E[q(W_{ij},\theta)]$ for all $\theta \in \Theta \backslash \{\theta_0\}$.

\noindent (ii) (a) $\theta \mapsto q\left(w, \theta\right)$ is Lipschitz with a universal Lipschitz constant.

\noindent\hspace{0.6cm}
(b) Each coordinate of $\theta \mapsto \nabla_{\theta}q(w,\theta)$ is Lipschitz with a universal Lipschitz constant.  

\noindent\hspace{0.6cm}
(c) Each coordinate of $\theta \mapsto \nabla_{\theta\theta^T}q(w,\theta)=\partial^2q\left(w,\theta\right)/\partial \theta \partial \theta^T$ is Lipschitz with a universal Lipschitz constant. 

\noindent (iii) $E[\sup_{\theta\in \Theta}q\left(W_{ij},\theta\right)]<\infty.$ 

\noindent (iv) $E\left[\sup_{\theta \in \Theta}\left\|\nabla_{\theta\theta^T}q\left(W_{ij},\theta\right)\right\|\right]<\infty.$ 

\noindent (v) $H=H\left(\theta^0\right)$ is nonsingular where $H(\theta)=-E\left[\nabla_{\theta\theta^T}q\left(W_{ij},\theta\right)\right].$

\noindent (vi)  $\dot q_{\sup}(\cdot)=\max_{r \in \left\{1,...,k\right\} }\left|\partial q(\cdot,\theta)/\partial \theta_r\right|$ satisfies $E[\dot q_{\sup}(W_{ij})^2]<\infty.$
\end{assumption}  

\noindent Assumption \ref{a:M} is analogous to the conditions required for Theorem 2.1 and Theorem 3.1 in \citet*{Newey1994}, which state the consistency and asymptotic normality, respectively, of the M-estimator under the conventional random sampling. 

We first state the consistency of the multiway algorithmic subsampling M-estimator.

\begin{lemma}[Consistency of the Multiway Algorithmic Subsampling M-estimator]\label{theorem:consistency_M}
If Assumptions \ref{a:sampling} and \ref{a:M} (i), (ii), (iii) hold,
then $\hat \theta \overset{P}{\rightarrow} \theta^0.$
\end{lemma}

\noindent
A proof is provided in Appendix \ref{sec:theorem:consistency_M}.
It follows from combining the arguments in the proof of \citet[][Theorem 2.1]{Newey1994} with our uniform weak law of large numbers for the multiway algorithmic subsampling (Lemma \ref{lemma:consistency}) presented in Section \ref{sec:wlln}.

We next state the asymptotic normality of the multiway algorithmic subsampling M-estimator.

\begin{theorem}[Asymptotic Normality of the Multiway Algorithmic Subsampling M-estimator]\label{theorem:asymptotic_normality_M}
If Assumptions \ref{a:sampling}, \ref{a:scalar} and \ref{a:M} hold, then $$\sqrt{\underline C}\left(\hat \theta-\theta^0\right)\overset{d}{\rightarrow}N\left(0,H^{-1}\Sigma H^{-1}\right),$$ where $H=-E\left[\nabla_{\theta\theta^T}q\left(W_{11},\theta^0\right)\right]$, $\Sigma=\Sigma_1+\Lambda \Sigma_2,$ $\Sigma_1=\lambda_1E\left[\nabla_{\theta}q\left(W_{11},\theta^0\right) \nabla_{\theta}q\left(W_{12},\theta^0\right)^T\right]+\lambda_2E\left[\nabla_{\theta}q\left(W_{11},\theta^0\right) \nabla_{\theta}q\left(W_{21},\theta^0\right)^T\right]$ and $\, \Sigma_2=E\left[\nabla_{\theta}q\left(W_{11},\theta^0\right) \nabla_{\theta}q\left(W_{11},\theta^0\right)^T\right].$
\end{theorem}

\noindent
A proof is provided in Appendix \ref{sec:theorem:asymptotic_normality_M}.
It follows from combining the arguments in the proof of \citet[][Theorem 3.1]{Newey1994} with our central limit theorem for the multiway algorithmic subsampling (Theorem \ref{theorem:clt}) presented in Section \ref{sec:clt}.

\subsection{Algorithmic Subsampling Variance Estimation}

\noindent The components, $H$ and $\Sigma$ in Theorem \ref{theorem:asymptotic_normality_M}, of the asymptotic variance of the multiway algorithmic subsampling M-estimator can be estimated by 
$$
\widetilde H=-\frac{1}{\hat L}\sum_{i=1}^{N}\sum_{j=1}^{M}Z_{ij}\nabla_{\theta\theta^T}q\left(W_{ij},\hat \theta\right)
$$ 
and 
$
\widetilde \Sigma=\widetilde \Sigma_1+\Lambda \widetilde \Sigma_2,
$
respectively,
where 
\begin{align*}
\widetilde \Sigma_1
=&
\frac{\underline{C}}{\widehat L^2}\sum_{i=1}^{N}\sum_{1\leq j,j'\leq M}Z_{ij}Z_{ij'}\nabla_{\theta}q\left(W_{ij},\hat \theta\right)\nabla_{\theta}q\left(W_{ij'},\hat \theta\right)^T\\
&+ \frac{\underline{C}}{\widehat L^2}\sum_{1\leq i, i'\leq N}\sum_{j=1}^MZ_{ij}Z_{i'j}\nabla_{\theta}q\left(W_{ij},\hat \theta\right)\nabla_{\theta}q\left(W_{i'j},\hat \theta\right)^T
\end{align*}
and 
$$
\widetilde \Sigma_2= \frac{1}{\hat L}\sum_{i=1}^{N}\sum_{j=1}^{M}Z_{ij}\nabla_{\theta}q\left(W_{ij},\hat \theta\right)\nabla_{\theta}q\left(W_{ij},\hat \theta\right)^T.
$$ 
Thus, we propose to estimate $H^{-1}\Sigma H^{-1}$ by the sample counterpart $\widetilde H^{-1}\widetilde \Sigma \widetilde H^{-1}.$ 
To guarantee that this asymptotic variance estimator works, we use the following assumption in addition.

\begin{assumption}\label{a:M:variance}
~\\
(i) $\theta \mapsto E\left[\nabla_{\theta\theta^T}q\left(W_{ij},\theta\right)\right]$ is continuous at $\theta^0.$

\noindent (ii) $\theta \mapsto \lambda_1E\left[\nabla_{\theta} q\left(W_{ij},\theta\right)\nabla_{\theta}q\left(W_{ij},\theta\right)^T\right]+\lambda_2E\left[\nabla_{\theta}q\left(W_{ij},\theta\right)\nabla_{\theta}q\left(W_{ij},\theta\right)^T\right]$
is continuous at $\theta^0.$

\noindent (iii) $\theta \mapsto E\left[\nabla_{\theta}q\left(W_{ij},\theta\right)\nabla_{\theta}q\left(W_{ij},\theta\right)^T\right]$ is continuous at $\theta^0.$
\end{assumption}
\noindent
With this additional assumption, $\widetilde H^{-1}\widetilde \Sigma \widetilde H^{-1}$ is consistent for the asymptotic variance $H^{-1}\Sigma H^{-1}$, as formally stated in the following theorem. 
\begin{theorem}[Consistency of the Asymptotic Variance of the Multiway Algorithmic Subsampling M-estimator]\label{theorem:consistency_M_variance}
If Assumptions \ref{a:sampling}, \ref{a:scalar}, \ref{a:M}, \ref{a:M:variance} hold,
then $\widetilde H^{-1}\widetilde \Sigma \widetilde H^{-1}$ is consistent for $H^{-1}\Sigma H^{-1}.$
\end{theorem}

A proof is provided in Appendix \ref{sec:theorem:consistency_M_variance}.
It follows by combining Lemma \ref{lemma:consistency} and similar lines of arguments to those in the proofs of Lemma \ref{theorem:consistency_M} and Theorem \ref{theorem:asymptotic_normality_M}.

\section{Simulation Studies}\label{sec:simulation}

As emphasized in Section \ref{sec:multiway_algorithmic_subsampling}, we discovered a new advantage of the algorithmic subsampling that it allows for robustness in inference against potential degeneracy of the asymptotic distribution under multiway clustering.
In this section, we use Monte Carlo simulations to demonstrate this robustness property.
Following \cite{menzel2017bootstrap}, we consider two broad categories of designs, namely additively separable designs (Section \ref{sec:separable}) and nonseparable designs (Section \ref{sec:nonseparable}).
For each of these two broad categories, we experiment with a design that leads to a non-degenerate asymptotic distribution and another design that leads to a degenerate asymptotic distribution if the algorithmic subsampling were not to be employed.
In total, we consider four designs.
The multiway algorithmic subsampling will be shown to yield more accurate finite sample coverage results than conventional methods robustly across all the four cases, thereby supporting the aforementioned theoretical discovery by this paper.

\subsection{Additively Separable Designs}\label{sec:separable}
First, we generate the two-way clustered array $\{Y_{ij}\}_{i \in [N], j \in [M]}$ according to the additively separable model
\begin{align*}
Y_{ij} = \sigma_a \alpha_i + \sigma_b \beta_j + \sigma_e \varepsilon_{ij},
\end{align*}
where $\beta_j$ and $\varepsilon_{ij}$ are i.i.d. standard normal, and $\alpha_i = (\zeta_i - \mu_\zeta)/\sigma_\zeta$ for $\log(\zeta_i) \stackrel{\text{i.i.d.}}{\sim} N(0,1)$, $\mu_\zeta = E[\zeta_i]$, and $\sigma_\zeta^2 = \text{Var}(\zeta_i)$.
With this basic setup, we consider two designs:
\begin{align*}
\text{Design 1}:& \ \sigma_a^2 = 0.5, \sigma_b^2 = 0.1, \text{ and }\sigma_e^2 = 0.2; \qquad\text{and}\\
\text{Design 2}:& \ \sigma_a^2 = 0.0, \sigma_b^2 = 0.0, \text{ and } \sigma_e^2 = 0.2.
\end{align*}
Note that Design 2, without $i$-specific randomness or $j$-specific randomness, would lead to a degenerate asymptotic distribution if the algorithmic subsampling were not employed.

Table \ref{tab:simulation_results_separable} reports simulation results for $N=M=40$, $80$, $160$, $320$, and $640$.
The top panel reports results for Design 1 (non-degenerate case), and the bottom panel reports results for Design 2 (degenerate case). 
Each panel contains results based on no algorithmic subsampling (i.e., $p=1$)\footnote{The 95\% coverage is computed based on our asymptotic variance formula as the special case with $p=1$.} and results based on the algorithmic subsampling (with the subsampling probabilities of $p=1\underline{C}/(NM)$ and $p=2\underline{C}/(NM)$) for estimation of the mean. 
The asymptotic variance is estimated using a random subsample of ten percent of the sample.
The displayed statistics are the bias (Bias), the standard deviation (SD), the root mean square error (RMSE), and the 95\% coverage (95\% Cover).

\begin{table}[t]
	\centering
	\scalebox{0.85}{
		\begin{tabular}{cccccccccccccccc}
		\multicolumn{16}{c}{\large Design 1: Non-Degenerate Case}\\
		\hline\hline
		&& \multicolumn{4}{c}{\large No Algorithmic Subsampling} && \multicolumn{9}{c}{\large Algorithmic Subsampling} \\
		\cline{8-16}
		&& \multicolumn{4}{c}{($p=1$)} && \multicolumn{4}{c}{$p=1\underline{C}/(NM)$} && \multicolumn{4}{c}{$p=2\underline{C}/(NM)$}\\
		\cline{3-6}\cline{8-11}\cline{13-16}
		$N$ & $M$ & Bias & SD & RMSE & 95\% && Bias & SD & RMSE & 95\% && Bias & SD & RMSE & 95\% \\
		\hline
			40 & 40 & 	0.006 & 0.127 & 0.127 & \bf 0.885 && 	0.005 & 0.194 & 0.194 & \bf 0.926 && 	-0.002 & 0.155 & 0.155 & \bf 0.908 \\
			80 & 80 & 	-0.001 & 0.087 & 0.087 & \bf 0.902 && 	0.002 & 0.131 & 0.131 & \bf 0.925 && 	-0.002 & 0.110 & 0.110 & \bf 0.916 \\
			160&160&	-0.001 & 0.060 & 0.060 & \bf 0.918 && 	0.000 & 0.094 & 0.094 & \bf 0.925 && 	0.002 & 0.079 & 0.079 & \bf 0.923 \\
			320&320&	0.001 & 0.044 & 0.044 & \bf 0.916 && 	0.000 & 0.068 & 0.068 & \bf 0.932 && 	0.000 & 0.057 & 0.057 & \bf 0.927\\
			640&640&	0.001 & 0.032 & 0.032 & \bf 0.922 && 	0.000 & 0.047 & 0.047 & \bf 0.948 && 	-0.001 & 0.040 & 0.040 & \bf 0.934\\
		\hline\hline
		\\
		\end{tabular}
	}
	\scalebox{0.85}{
		\begin{tabular}{cccccccccccccccc}
		\multicolumn{16}{c}{\large Design 2: Degenerate Case}\\
		\hline\hline
		&& \multicolumn{4}{c}{\large No Algorithmic Subsampling} && \multicolumn{9}{c}{\large Algorithmic Subsampling} \\
		\cline{8-16}
		&& \multicolumn{4}{c}{($p=1$)} && \multicolumn{4}{c}{$p=1\underline{C}/(NM)$} && \multicolumn{4}{c}{$p=2\underline{C}/(NM)$}\\
		\cline{3-6}\cline{8-11}\cline{13-16}
		$N$ & $M$ & Bias & SD & RMSE & 95\% && Bias & SD & RMSE & 95\% && Bias & SD & RMSE & 95\%\\
		\hline
				40 & 40 & 	0.000 & 0.011 & 0.011 & \bf 0.999 && 	-0.001 & 0.105 & 0.105 & \bf 0.981 && 	-0.001 & 0.049 & 0.049 & \bf 0.986\\
				80 & 80 & 	0.000 & 0.006 & 0.006 & \bf 1.000 && 	-0.001 & 0.051 & 0.051 & \bf 0.963 && 	0.001 & 0.036 & 0.036 & \bf 0.970\\
				160&160&	0.000 & 0.003 & 0.003 & \bf 1.000 && 	0.001 & 0.036 & 0.036 & \bf 0.959 && 	0.000 & 0.026 & 0.026 & \bf 0.961\\
				320&320&	0.000 & 0.001 & 0.001 & \bf 1.000 && 	0.000 & 0.025 & 0.025 & \bf 0.959 && 	0.000 & 0.018 & 0.018 & \bf 0.960\\
				640&640&	0.000 & 0.001 & 0.001 & \bf 1.000 && 	0.000 & 0.018 & 0.018 & \bf 0.944 && 	0.000 & 0.013 & 0.013 & \bf 0.950\\
		\hline\hline
		\end{tabular}
	}
	\caption{Simulation results for the additively separable design with $N=M=40$, $80$, $160$, $320$, and $640$ based on 2,500 Monte Carlo iterations. The top panel reports results for Design 1 (non-degenerate case), whereas the bottom panel reports results for Design 2 (degenerate case). Each panel contains results based on no algorithmic subsampling ($p=1$), results based on the algorithmic subsampling with $p=1\underline{C}/(NM)$, and results based on the algorithmic subsampling with $p=2\underline{C}/(NM)$ for estimation of the mean. The displayed statistics are the bias (Bias), the standard deviation (SD), the root mean square error (RMSE), and the 95\% coverage (95\%).}${}$
	\label{tab:simulation_results_separable}
\end{table}

Observe that the 95\% coverage frequencies are closer to the nominal probability of 95\% with a use of the algorithmic subsampling than without a use of it.
This observation is robustly true in both Design 1 (non-degenerate case) and Design 2 (degenerate case).
For Design 2 or the degenerate case, in particular, the coverage frequency moves away from the nominal probability as the sample size increases if the algorithmic subsampling were not used.
On the other hand, the coverage frequency approaches the nominal probability as the sample size increase if the algorithmic subsampling is used.
These results demonstrate the aforementioned robustness property of the multiway algorithmic subsampling against potential degeneracy of the asymptotic distribution.
We also experimented with additional simulation settings with much larger $N$ and $M$ and other subsampling probabilities for the algorithmic subsampling variance estimation, but we observe the same qualitative patterns in the results under these alternative settings.

On the one hand, $p=1$ leads to more precision, as quantified by smaller RMSE.
On the other hand, $p=1$ leads to larger coverage as observed above.
These two phenomena may appear contradictory at first glance. 
The relevant issues are with the variance estimation, and not with the point estimates.
These results precisely highlight the cases of degeneracy. 
The asymptotic normality with the $\sqrt{\underline{C}}$-rate fails under the degeneracy if we do not use the algorithmic subsampling, i.e., if $p=1$. 
Therefore, the standard errors are misleadingly larger compared to the actual RMSE of the estimator and the simulated coverage rates exceed the nominal coverage probability in the degenerate case with $p=1$. This is the main reason why we propose to use the algorithmic subsampling (i.e., $p < 1$) to have inference with estimated variance robust against the degeneracy. 

\subsection{Nonseparable Designs}\label{sec:nonseparable}
Second, we generate the two-way clustered array $\{Y_{ij}\}_{i \in [N], j \in [M]}$ according to the non-additive model
\begin{align*}
Y_{ij} = (\alpha_i-\mu_a)(\beta_j-\mu_b) - \mu_a \mu_b + \varepsilon_{ij},
\end{align*}
where $\alpha_i$, $\beta_j$ and $\varepsilon_{ij}$ are i.i.d. standard normal.
With this basic setup, we consider two designs:
\begin{align*}
\text{Design 3}:& \ \mu_a=1.0, \text{ and } \mu_b=1.0; \qquad\text{and}\\
\text{Design 4}:& \ \mu_a=0.0, \text{ and } \mu_b=0.0.
\end{align*}
Note that Design 4 would lead to a degenerate asymptotic distribution that is a  Gaussian chaos, which is non-Gaussian \citep[cf.][]{menzel2017bootstrap}, if the algorithmic subsampling were not employed.

Table \ref{tab:simulation_results_nonseparable} reports simulation results for $N=M=40$, $80$, $160$, $320$, and $640$.
The top panel reports results for Design 3 (non-degenerate case), and the bottom panel reports results for Design 4 (degenerate case). 
Each panel contains results based on no algorithmic subsampling (i.e., $p=1$) and results based on the algorithmic subsampling (with the subsampling probabilities of $p=1\underline{C}/(NM)$ and $p=2\underline{C}/(NM)$) for estimation of the mean. 
The asymptotic variance is estimated using a random subsample of ten percent of the sample.
The displayed statistics are the bias (Bias), the standard deviation (SD), the root mean square error (RMSE), and the 95\% coverage (95\% Cover).

\begin{table}[t]
	\centering
	\scalebox{0.85}{
		\begin{tabular}{cccccccccccccccc}
		\multicolumn{16}{c}{\large Design 3: Non-Degenerate Case}\\
		\hline\hline
		&& \multicolumn{4}{c}{\large No Algorithmic Subsampling} && \multicolumn{9}{c}{\large Algorithmic Subsampling} \\
		\cline{8-16}
		&& \multicolumn{4}{c}{($p=1$)} && \multicolumn{4}{c}{$p=1\underline{C}/(NM)$} && \multicolumn{4}{c}{$p=2\underline{C}/(NM)$}\\
		\cline{3-6}\cline{8-11}\cline{13-16}
		$N$ & $M$ & Bias & SD & RMSE & 95\% && Bias & SD & RMSE & 95\% && Bias & SD & RMSE & 95\% \\
		\hline
			40 & 40 & 	-0.002 & 0.221 & 0.221 & \bf 0.943 && 	0.007 & 0.382 & 0.382 & \bf 0.955 && 	0.003 & 0.316 & 0.315 & \bf 0.944\\
			80 & 80 & 	0.008 & 0.160 & 0.160 & \bf 0.943 && 	-0.004 & 0.272 & 0.272 & \bf 0.949 && 	0.003 & 0.221 & 0.221 & \bf 0.956\\
			160&160&	-0.002 & 0.113 & 0.113 & \bf 0.942 && 	0.001 & 0.190 & 0.190 & \bf 0.952 && 	0.001 & 0.158 & 0.158 & \bf 0.945\\
			320&320&	0.000 & 0.079 & 0.079 & \bf 0.944 && 	0.004 & 0.140 & 0.140 & \bf 0.940 && 	0.004 & 0.110 & 0.110 & \bf 0.953\\
			640&640&	-0.002 & 0.055 & 0.055 & \bf 0.956 && 	-0.001 & 0.096 & 0.096 & \bf 0.954 && 	0.001 & 0.079 & 0.079 & \bf 0.949\\
		\hline\hline
		\\
		\end{tabular}
	}
	\scalebox{0.85}{
		\begin{tabular}{cccccccccccccccc}
		\multicolumn{16}{c}{\large Design 4: Degenerate Case}\\
		\hline\hline
		&& \multicolumn{4}{c}{\large No Algorithmic Subsampling} && \multicolumn{9}{c}{\large Algorithmic Subsampling} \\
		\cline{8-16}
		&& \multicolumn{4}{c}{($p=1$)} && \multicolumn{4}{c}{$p=1\underline{C}/(NM)$} && \multicolumn{4}{c}{$p=2\underline{C}/(NM)$}\\
		\cline{3-6}\cline{8-11}\cline{13-16}
		$N$ & $M$ & Bias & SD & RMSE & 95\% && Bias & SD & RMSE & 95\% && Bias & SD & RMSE & 95\% \\
		\hline
			40 & 40 & 	0.000 & 0.036 & 0.036 & \bf 1.000 && 	-0.004 & 0.221 & 0.221 & \bf 0.980 && 	-0.004 & 0.160 & 0.160 & \bf 0.980\\
			80 & 80 & 	0.000 & 0.018 & 0.018 & \bf 1.000 && 	-0.001 & 0.155 & 0.155 & \bf 0.971 && 	0.001 & 0.110 & 0.110 & \bf 0.975\\
			160&160&	0.000 & 0.009 & 0.009 & \bf 1.000 && 	0.002 & 0.113 & 0.113 & \bf 0.956 && 	-0.001 & 0.078 & 0.078 & \bf 0.966\\
			320&320&	0.000 & 0.005 & 0.005 & \bf 0.999 && 	0.000 & 0.078 & 0.078 & \bf 0.956 && 	0.000 & 0.057 & 0.057 & \bf 0.952\\
			640&640&	0.000 & 0.002 & 0.002 & \bf 0.999 && 	0.001 & 0.057 & 0.057 & \bf 0.946 && 	-0.001 & 0.039 & 0.039 & \bf 0.952\\
		\hline\hline
		\end{tabular}
	}
	\caption{Simulation results for the nonseparable design with $N=M=40$, $80$, $160$, $320$, and $640$ based on 2,500 Monte Carlo iterations. The top panel reports results for Design 3 (non-degenerate case), whereas the bottom panel reports results for Design 4 (degenerate case). Each panel contains results based on no algorithmic subsampling ($p=1$), results based on the algorithmic subsampling with $p=1\underline{C}/(NM)$, and results based on the algorithmic subsampling with $p=2\underline{C}/(NM)$ for estimation of the mean. The displayed statistics are the bias (Bias), the standard deviation (SD), the root mean square error (RMSE), and the 95\% coverage (95\%).}${}$
	\label{tab:simulation_results_nonseparable}
\end{table}

Similarly to the case with the additively separable design, observe that the 95\% coverage frequencies are closer to the nominal probability of 95\% with a use of the algorithmic subsampling than without a use of it.
This observation is robustly true in both Design 3 (non-degenerate case) and Design 4 (degenerate case).
For Design 4 or the degenerate case, in particular, the coverage frequency moves away from the nominal probability as the sample size increases if the algorithmic subsampling were not used.
On the other hand, the coverage frequency approaches the nominal probability as the sample size increase if the algorithmic subsampling is used.
As before, these results demonstrate the aforementioned robustness property of the multiway algorithmic subsampling against potential degeneracy of the asymptotic distribution.
We also experimented with additional simulation settings with much larger $N$ and $M$ and other subsampling probabilities for the algorithmic subsampling variance estimation, but we observe the same qualitative patterns in the results under these alternative settings.

\section{Application to Scanner Data}\label{sec:empirical}

In this section, we demonstrate an application of our proposed method to an analysis of demand for differentiated products using scanner data from the Dominick's Finer Foods (DFF) retail chain.\footnote{We thank James M. Kilts Center, University of Chicago Booth School of Business for allowing us to use this data set. It is available at https://www.chicagobooth.edu/research/kilts/datasets/dominicks.}
Scanner data may be subject to two-way cluster dependence, as mentioned in Section \ref{sec:introduction}.
Specifically, common demand shocks within a market may induce statistical dependence among different products within that market. 
Similarly, common supply shocks by a producer may induce statistical dependence among different markets within the product produced by that producer.
In this light, a researcher would like to use a two-way cluster robust variance estimate for inference about the model parameters.
However, the scanner data from the Dominick's Finer Foods (DFF) retail chain are too large, and today's computational resources will not permit the two-way cluster robust variance estimation in reasonable lengths of time.
A simple way to overcome this problem is to use the full sample for parameter estimation and to use a subsample for variance estimation, but this approach fails to deliver robustly valid inference.
Hence, we use our proposed multiway algorithmic subsampling method for estimation and two-way cluster robust inference about the key demand model parameter.

Following the literature \citep[for instance, see a survey by][]{nevo2000practitioner} on analysis of demand for differentiated products with an additive Type-I-Extreme-Value error, we use the GMM approach with the moment restriction
\begin{align}\label{eq:moment_demand}
g(W_{ij},\theta) = \zeta_{ij} (\ln(S_{ij}) - \ln(S_{0j}) - \ln(P_{ij})\theta_1 - X_{ij}^T\theta_{-1}),
\end{align}
where 
$i$ indexes products (universal product code, hereafter referred to as UPC),
$j$ indexes markets (store $\times$ week),
$S_{ij}$ denotes the share of product $i$ in market $j$,
$P_{ij}$ denotes the price,
$X_{ij}$ denotes a vector of controls (the UPC fixed effects and a time trend),
$\zeta_{ij}$ denotes instruments,
and $W_{ij} = (S_{ij},P_{ij},X_{ij}^T,\zeta_{ij}^T)^T$.\footnote{In case where the model involves product fixed effects, the algorithmic subsampling can be applied to within-transformation. This operation incurs additional computational costs, although this is a common issue in fixed-effect methods in general. In case a model involves two-way fixed effects, two-way differencing may induce a more complicated dependence structure especially under unbalanced panels. An alternative approach may be to use instrumental variables. We leave rigorous treatments of such a variety of extensions to fixed-effect models for future research.} 
In addition to the elements in $X_{ij}$, the instrument vector includes $\zeta_{ij}$ as an excluded variable the wholesale costs, which are calculated by inverting the gross margin.
We drop those observations for which $\ln(S_{ij}) - \ln(S_{0j})$ is not finite,\footnote{\label{foot:drop}In other words, we drop observations with the zero market share. Dropping these observations may generally incur a trimming bias. We adopt this trimming as it is a standard practice in the literature of demand analysis for differentiated products markets, and we consider the possibly biased estimand as our pseudo-true value.} as well as those observations with missing values.
The parameter vector in the model consists of $\theta = (\theta_1,\theta_{-1}^T)^T$, and we are in particular interested in the price coefficient $\theta_1$.

We consider four product categories: beer, oats, snacks, and canned tuna.
Table \ref{tab:summary} summarizes the sizes of the original data in terms of various dimensions.
It first shows the number of UPCs, the number of weeks, and the number of stores for each product category.
As we define a product as that identified by the UPC, the number of products $N$ coincides with the number of UPCs.
We define a market as the unique combination of the week and the store.
Therefore, the number of markets $M$ is close to, but is generally smaller than, the product of the number of weeks and the number of stores.
It is smaller than the na\"ive product because of the unbalancedness in data.
Finally, the bottom row shows the total number of observations, which is again smaller than the na\"ive product $NM$ because of the unbalancedness in data.

\begin{table}[t]
	\centering
		\begin{tabular}{lcccc}
		\hline\hline
		                         & Beer    & Oats & Snacks & Tuna\\
		\hline
			Number of UPCs         & 788     & 96      & 425     & 94\\
			Number of Weeks        & 303     & 306     & 386     & 375\\
			Number of Stores       & 89      & 93      & 94      & 93\\
			Number of Products $N$ & 788     & 96      & 425     & 94\\
			Number of Markets $M$  & 22,299  & 26,210  & 32,708  & 31,853\\
			Number of Observations &3,990,672&1,333,465&5,427,491&1,048,575\\
		\hline\hline
		\end{tabular}
	\caption{Data sizes of the four product categories: beer, oats, snacks, and canned tuna.}${}$
	\label{tab:summary}
\end{table}

We now apply our multiway algorithmic subsampling GMM with the moment function defined in \eqref{eq:moment_demand} for each of the four product categories.
Table \ref{tab:estimation_results} summarizes the estimation results. 
The table displays the probability $p$ of algorithmic subsampling, the corresponding estimates and their standard errors for the price coefficient, and computational time in seconds for each of parameter estimation and asymptotic variance estimation.

\begin{table}[t]
	\centering
	\scalebox{0.89}{
		\begin{tabular}{lcccccccc}
		\hline\hline
		                         & Beer  & Beer  & Oats  & Oats  &Snacks&Snacks& Tuna & Tuna\\
		\cline{2-9}
	 	$p$ & $\frac{100\underline{C}}{NM}$ & $\frac{200\underline{C}}{NM}$ & $\frac{100\underline{C}}{NM}$ & $\frac{200\underline{C}}{NM}$ & $\frac{100\underline{C}}{NM}$ & $\frac{200\underline{C}}{NM}$ & $\frac{100\underline{C}}{NM}$ &$\frac{200\underline{C}}{NM}$ \\
			  & 0.004 & 0.009 & 0.004 & 0.008 & 0.003 & 0.006 & 0.003 & 0.006\\
		\hline
	Price Coefficient $^\dagger$&-0.223$^{\ast\ast}$ &-0.334$^{\ast\ast\ast}$ &-1.186$^{\ast\ast\ast}$ &-1.273$^{\ast\ast\ast}$ &-1.155$^{\ast\ast\ast}$ &-1.105$^{\ast\ast\ast}$ &-1.605$^{\ast}$ &-0.936$^{\ast}$\\
			                       &(0.102)&(0.066)&(0.173)&(0.103)&(0.159)&(0.151)&(0.985)&(0.500)\\
Computational Time $^\ddagger$\\\cline{1-1}
Parameter Estimation & 7.313 & 13.952 & 0.081 & 0.129 & 2.525 & 5.582 & 0.063 & 0.092\\
Variance Estimation & 1223 & 4458 & 34 & 189 & 676 & 2901 & 10 & 46\\
		\hline\hline
		\end{tabular}
	}
	\caption{Results of the estimation of the price coefficient. $^\dagger$ The standard errors are shown in parentheses under the estimates. ***\ p$<$0.01, **\ p$<$0.05, *\ p$<$0.10. $^\ddagger$ Computational time is expressed in seconds based on a single processor of Intel Xeon Processor E5-2687W V4.}${}$
	\label{tab:estimation_results}
\end{table}

First, observe that the estimates of the price coefficient are negative, as expected, and are statistically significant at the level of 95\% for each column except for tuna despite efficiency loss due to the algorithmic subsampling and despite the two-way cluster robustness in the asymptotic variance.
As emphasized in Sections \ref{sec:multiway_algorithmic_subsampling} and \ref{sec:simulation}, the algorithmic subsampling with $p \propto \underline{C}/(NM)$ allows these standard errors to have asymptotically accurate coverage robustly against potential degeneracy, unlike the conventional two-way cluster robust standard errors without the algorithmic subsampling.

Second, the computational time for parameter estimation is within about a dozen of seconds for each column, given that the algorithmic subsampling extracts only the proportions, $p \approx 0.003--0.009$, of the original
sample sizes. 
However, it is the asymptotic variance estimation that costs more computational time under multiway cluster dependence. 
Focusing on the beer product category, for instance, even the
algorithmic subsampling that extracts only the $p \approx 0.004$ portion of the original sample size requires 1223 seconds of computation for variance estimation. When the proportion doubles to $p \approx 0.009$, then the computational time nearly quadruples to 4458 seconds. A na\"ive calculation implies that the use of the full sample without the algorithmic subsampling would require about three years. 

\section{Conclusion}\label{sec:conclusion}
In this paper, we propose a novel method of algorithmic subsampling for multiway cluster dependent data.
We develop asymptotic statistical properties of this proposed method.
Specifically, we develop a new uniform weak law of large numbers and a new central limit theorem for the multiway algorithmic subsample means.
As a consequence of the new central limit theorem, we show that the algorithmic subsampling allows for robustness against potential degeneracy of the asymptotic distribution under multiway clustering at the cost of efficiency and power loss due to the algorithmic subsampling.
Applying these basic asymptotic statistical theories, we derive the consistency and the asymptotic normality for the multiway algorithmic subsampling generalized method of moments estimator and for the multiway algorithmic subsampling M-estimator.

Our main finding that the algorithmic subsampling allows for the robustness against degeneracy in the asymptotic distribution is novel in the literature on multiway clustering.
Indeed, the method of inference by \citet*{mackinnon2019wild} as well as \citet*{cameron2012robust} adapts to the Gaussian degeneracy. 
However, these existing methods do not adapt to the class of non-Gaussian degenerate asymptotic distributions.
In contrast, the asymptotic distribution under the algorithmic subsampling adapts even to the non-Gaussian degeneracy as well.
The bootstrap method of \citet*{menzel2017bootstrap} is robust against the non-Gaussian degeneracy.
Our proposed method via the algorithmic subsampling leads to the exact limit distribution, and thus non-conservative inference, unlike the method of \citet*{menzel2017bootstrap}.
With these said, we once again emphasize that these merits come at the cost of efficiency and power loss by disposing parts of big data.

Finally, we shed some light on possible future directions. In this paper we consider non-nested multiway clustering \citep*[as in][]{cameron2012robust}. In practice, the researcher may be interested in applications with nested clustering in one or more cluster dimensions. Under the current framework, one could take the coarsest levels of clustering.  Handling it in a more efficient way is a useful topic but is out of the scope of this paper. In addition, in \citet*{mackinnon2020testing}, formal theory is developed for testing the correct level of (one-way) clustering. One could consider to generalize such test for multiway nested clustering, which is also left for future research.

\begin{appendix}
\section*{Appendix}

Throughout this appendix, for any arrays $(a_{NM})$ and $(b_{NM})$, denote $a_{NM} \lesssim b_{NM}$ for $a_{NM}\le C b_{NM}$ for some positive constant $C$ independent of sample size.

\section{Choice of the Subsample Size}\label{sec:choice}

Theorem \ref{theorem:clt} provides a guidance on rates at which $p$ should converge in order to guarantee the robustness against degeneracy under multiway cluster sampling.
Specifically, $p = p_{MN}$ should be chosen so that $\Lambda = \lim_{N,M \rightarrow \infty} (\underline{C}/(NM))((1-p)/p) > 0$ holds.
To this goal, it is in particular sufficient to choose
\begin{align*}
 p = c\frac{\underline{C}}{NM}
\text{ for some } c > 0.
\end{align*}
For our asymptotic properties with the robustness, any choice of a positive constant $c$ works in theory.
Simulation studies presented in Section \ref{sec:simulation} demonstrate that even the na\"ive choices, such as $c=1$ and $c=2$, result in excellent finite-sample performances across various alternative data generating designs.

That said, it is also useful as well to provide a data-driven method to choose $c$ based on a well-defined criterion.
In this section, we propose a method to this end following the idea of power analysis which is often employed to determine experimental sample size.
Suppose that a researcher has in mind a maximum tolerable level $V_{\max}$ of the approximate variance $\Gamma/\underline{C}$ of $\hat L^{-1}\sum_{i=1}^{N}\sum_{j=1}^{M}Z_{ij}f\left(W_{ij}\right)$ in the asymptotic normal approximation by Theorem \ref{theorem:clt}.

First, choose a preliminary positive value of $c^{\text{pre}}$, set $p^{\text{pre}} = c^{\text{pre}} \underline{C}/(NM)$, generate i.i.d Bernoulli$\left(p^{\text{pre}}\right)$ random variables $\{Z^{\text{pre}}_{ij}: 1 \le i \le N, 1 \le j \le M\}$ independently from data, and set $\hat L^{\text{pre}} = \sum_{i=1}^N \sum_{j=1}^M Z^{\text{pre}}_{ij}$.
Then, estimate $\Gamma_A=\lambda_1 E\left[f\left(W_{11}\right)f^T\left(W_{12}\right)\right]+\lambda_2 E\left[f\left(W_{11}\right)f^T\left(W_{21}\right)\right]$ and $\Gamma_B=E\left[f\left(W_{11}\right)f^T\left(W_{11}\right)\right]$ by
$$
\widehat \Gamma^{\text{pre}}_A=\frac{\underline{C}}{(\widehat L^{\text{pre}})^2}\sum_{i=1}^{N}\sum_{1\leq j,j'\leq M}Z^{\text{pre}}_{ij}Z^{\text{pre}}_{ij'}f\left(W_{ij}\right)f\left(W_{ij'}\right)+ \frac{\underline{C}}{ (\widehat L^{\text{pre}})^2}\sum_{1\leq i, i'\leq N}\sum_{j=1}^MZ^{\text{pre}}_{ij}Z^{\text{pre}}_{i'j}f\left(W_{ij}\right)f\left(W_{i'j}\right)
$$ 
and 
$$
\widehat \Gamma^{\text{pre}}_B= \frac{1}{\hat L^{\text{pre}}}\sum_{i=1}^{N}\sum_{j=1}^{M}Z^{\text{pre}}_{ij}f\left(W_{ij}\right)f\left(W_{ij}\right),
$$ 
respectively.
Finally, solve
\begin{align*}
\underline{C} V_{\max} 
=
\widehat \Gamma^{\text{pre}}_A + \frac{\underline{C}}{NM} \frac{NM-c\underline{C}}{c\underline{C}} \widehat \Gamma^{\text{pre}}_B
\end{align*}
for $c$ to find the value $c^\ast$ of $c$.
This plug-in procedure yields the subsample size rate  $p^{\text{pre}} = c^{\text{pre}} \underline{C}/(NM)$, under which the approximate variance $\Gamma/\underline{C}$ of $\hat L^{-1}\sum_{i=1}^{N}\sum_{j=1}^{M}Z_{ij}f\left(W_{ij}\right)$ is close to the target level $V_{\max}$.
Note the similarity of this procedure to the power analysis for sample size calculation, which is often employed by experimental researchers.

We remark that this proposed procedure of choosing the subsample size differs from that proposed by \citet[][Section 7.2]{LeeNg2020ARE}.
This difference in the approaches taken is due to the different goals under different dependence structures.
\citet{LeeNg2020ARE} base their requirement for the subsample size on a condition that guarantees the subspace embedding \citep[][Definition 1]{LeeNg2020ARE}.
On the other hand, we base our requirement for the subsample size on attaining the robustness against degeneracy under multiway cluster sampling.

While the simulation studies presented in Section \ref{sec:simulation} are based fixed $c \in \{1,2\}$, we now present simulation results under the above choice rule of the subsample size.
We set $V_{\max} = 0.5/\underline{C}$ throughout, and start with $c=1$ for preliminary estimation of $\widehat\Gamma_A^{\text{pre}}$ and $\widehat\Gamma_B^{\text{pre}}$.
Focusing on the degenerate case, Table \ref{tab:simulation_results_separable_datadrivenp} summarizes simulation results under the additively separable designs, as the counterpart of Table \ref{tab:simulation_results_separable} in the main text.
Similarly, focusing on the degenerate case, Table \ref{tab:simulation_results_nonseparable_datadrivenp} summarizes simulation results under the nonseparable designs, as the counterpart of Table \ref{tab:simulation_results_nonseparable} in the main text.
Overall, we observe qualitatively similar patterns here to those presented in the main text.

\begin{table}
	\centering
	\scalebox{0.85}{
		\begin{tabular}{cccccccccccccccc}
		\multicolumn{11}{c}{\large Design 2: Degenerate Case}\\
		\hline\hline
		&& \multicolumn{4}{c}{\large No Algorithmic Subsampling} && \multicolumn{4}{c}{\large Algorithmic Subsampling} \\
		\cline{3-6}\cline{8-11}
		&& \multicolumn{4}{c}{($p=1$)} && \multicolumn{4}{c}{$p = c^\ast \underline{C}/(NM)$} \\
		\cline{3-6}\cline{8-11}
		$N$ & $M$ & Bias & SD & RMSE & 95\% && Bias & SD & RMSE & 95\% \\
		\hline
				40 & 40 & 	0.000 & 0.011 & 0.011 & \bf 0.999 && 0.001 & 0.071 & 0.071 & \bf 0.987\\
				80 & 80 & 	0.000 & 0.006 & 0.006 & \bf 1.000 && 0.001 & 0.050 & 0.050 & \bf 0.972\\
				160&160&	0.000 & 0.003 & 0.003 & \bf 1.000 &&-0.001 & 0.036 & 0.036 & \bf 0.964\\
				320&320&	0.000 & 0.001 & 0.001 & \bf 1.000 & &0.000 & 0.025 & 0.025 & \bf 0.953\\
				640&640&	0.000 & 0.001 & 0.001 & \bf 1.000 && 0.001 & 0.017 & 0.017 & \bf 0.955\\
		\hline\hline
		\end{tabular}
	}
	\caption{Simulation results for the additively separable design with $N=M=40$, $80$, $160$, $320$ and $640$ based on 2,500 Monte Carlo iterations.
	Each panel contains results based on no algorithmic subsampling ($p=1$) and results based on the algorithmic subsampling with $p = c^\ast \underline{C}/(NM)$, for estimation of the mean. The displayed statistics are the bias (Bias), the standard deviation (SD), the root mean square error (RMSE), and the 95\% coverage (95\%).}${}$
	\label{tab:simulation_results_separable_datadrivenp}
\end{table}

\begin{table}
	\centering
	\scalebox{0.85}{
		\begin{tabular}{cccccccccccccccc}
		\multicolumn{11}{c}{\large Design 4: Degenerate Case}\\
		\hline\hline
		&& \multicolumn{4}{c}{\large No Algorithmic Subsampling} && \multicolumn{4}{c}{\large Algorithmic Subsampling} \\
		\cline{3-6}\cline{8-11}
		&& \multicolumn{4}{c}{($p=1$)} && \multicolumn{4}{c}{$p = c^\ast \underline{C}/(NM)$} \\
		\cline{3-6}\cline{8-11}
		$N$ & $M$ & Bias & SD & RMSE & 95\% && Bias & SD & RMSE & 95\% \\
		\hline
			40 & 40 & 	0.000 & 0.036 & 0.036 & \bf 1.000 && 0.000 & 0.084 & 0.084 & \bf 0.990\\
			80 & 80 & 	0.000 & 0.018 & 0.018 & \bf 1.000 && 0.000 & 0.067 & 0.067 & \bf 0.984\\
			160&160&	0.000 & 0.009 & 0.009 & \bf 1.000 &&-0.001 & 0.051 & 0.051 & \bf 0.966\\
			320&320&	0.000 & 0.005 & 0.005 & \bf 0.999 && 0.001 & 0.038 & 0.038 & \bf 0.957\\
			640&640&	0.000 & 0.002 & 0.002 & \bf 0.999 &&-0.000 & 0.027 & 0.027 & \bf 0.958\\
		\hline\hline
		\end{tabular}
	}
	\caption{Simulation results for the nonseparable design with $N=M=40$, $80$, $160$, $320$ and $640$ based on 2,500 Monte Carlo iterations. 
	Each panel contains results based on no algorithmic subsampling ($p=1$) and results based on the algorithmic subsampling with $p = c^\ast \underline{C}/(NM)$, for estimation of the mean. The displayed statistics are the bias (Bias), the standard deviation (SD), the root mean square error (RMSE), and the 95\% coverage (95\%).}${}$
	\label{tab:simulation_results_nonseparable_datadrivenp}
\end{table}

\section{Additional Simulation Results}

The simulation studies presented in Section \ref{sec:simulation} in the main text and Appendix \ref{sec:choice} compare inference results only across those methods that assume two-way clustering.
This section extends these simulation analyses by comparing the finite-sample performance of our proposed method with more conventional methods that assume i.i.d. sampling and one-way clustering as well as two-way clustering.

We continue to use the same simulation designs from Section \ref{sec:simulation}.
Namely, data are generated according to Designs 1--4.
Sample sizes are varied as $N=M=$ 40, 80, 160, 320 and 640.
Each set of simulations consists of 2,500 Monte Carlo iterations.
Unlike Section \ref{sec:simulation}, however, we also compute 95\% coverage frequencies with the Eicker-Huber-White robust variance estimator (0-Way Cluster) and the conventional one-way cluster-robust variance estimator (1-Way Cluster) in addition to the two-way cluster-robust variance estimator (2-Way Custer).
Tables \ref{tab:simulation_results_separable1_zeroonetwo} and \ref{tab:simulation_results_separable2_zeroonetwo} summarize the results for Designs 1--2 and Designs 3--4, respectively.

\begin{table}
	\centering
	\scalebox{1}{
\begin{tabular}{cccccccccc}
\multicolumn{10}{c}{Design 1: Non-Degenerate Case}\\
\hline\hline
&& \multicolumn{5}{c}{} && \multicolumn{2}{c}{Algorithmic}\\
&& \multicolumn{5}{c}{No Algorithmic Subsampling} && \multicolumn{2}{c}{Subsampling}\\
\cline{3-7}\cline{9-10}
&& 0-Way && 1-Way && \multicolumn{4}{c}{2-Way Cluster}\\
\cline{7-10}
$N$ & $M$ & Cluster && Cluster && $p=1$ && $p=1\underline{C}/({NM})$ & $p=2\underline{C}/({NM})$ \\
\hline
40 & 40 &     0.270 && 0.834 && 0.885 && 0.926 & 0.908\\
80 & 80 &     0.192 && 0.837 && 0.902 && 0.925 & 0.916\\
160 & 160 & 0.140 && 0.854 && 0.918 && 0.925 & 0.923\\
320 & 320 & 0.097 && 0.846 && 0.916 && 0.932 & 0.927\\
640 & 640 & 0.074 && 0.845 && 0.922 && 0.948 & 0.934\\
\hline\hline
\end{tabular}
	}

\scalebox{1}{
\begin{tabular}{cccccccccc}
\multicolumn{10}{c}{Design 2: Degenerate Case}\\
\hline\hline
&& \multicolumn{5}{c}{} && \multicolumn{2}{c}{Algorithmic}\\
&& \multicolumn{5}{c}{No Algorithmic Subsampling} && \multicolumn{2}{c}{Subsampling}\\
\cline{3-7}\cline{9-10}
&& 0-Way && 1-Way && \multicolumn{4}{c}{2-Way Cluster}\\
\cline{7-10}
$N$ & $M$ & Cluster && Cluster && $p=1$ && $p=1\underline{C}/({NM})$ & $p=2\underline{C}/({NM})$ \\
\hline
40 & 40 &     0.952 && 1.000 && 0.999 && 0.981 & 0.986\\
80 & 80 &     0.960 && 1.000 && 1.000 && 0.963 & 0.970\\
160 & 160 & 0.945 && 1.000 && 1.000 && 0.959 & 0.961\\
320 & 320 & 0.948 && 1.000 && 1.000 && 0.959 & 0.960\\
640 & 640 & 0.951 && 1.000 && 1.000 && 0.944 & 0.950\\
\hline\hline
\end{tabular}
	}
	\caption{95\% coverage frquencies of various inference methods in the additively separable designs with  $N=M=40$, $80$, $160$, $320$ and $640$ based on 2,500 Monte Carlo iterations.}${}$
	\label{tab:simulation_results_separable1_zeroonetwo}
\end{table}
\begin{table}
	\centering
	\scalebox{1}{
\begin{tabular}{cccccccccc}
\multicolumn{10}{c}{Design 3: Non-Degenerate Case}\\
\hline\hline
&& \multicolumn{5}{c}{} && \multicolumn{2}{c}{Algorithmic}\\
&& \multicolumn{5}{c}{No Algorithmic Subsampling} && \multicolumn{2}{c}{Subsampling}\\
\cline{3-7}\cline{9-10}
&& 0-Way && 1-Way && \multicolumn{4}{c}{2-Way Cluster}\\
\cline{7-10}
$N$ & $M$ & Cluster && Cluster && $p=1$ && $p=1\underline{C}/({NM})$ & $p=2\underline{C}/({NM})$ \\
\hline
40 & 40 &     0.341 && 0.914 && 0.943 && 0.955 & 0.944\\
80 & 80 &     0.240 && 0.927 && 0.943 && 0.949 & 0.956\\
160 & 160 & 0.170 && 0.924 && 0.942 && 0.952 & 0.945\\
320 & 320 & 0.122 && 0.916 && 0.944 && 0.940 & 0.953\\
640 & 640 & 0.089 && 0.923 && 0.956 && 0.954 & 0.949\\
\hline\hline
\end{tabular}
	}

\scalebox{1}{
\begin{tabular}{cccccccccc}
\multicolumn{10}{c}{Design 4: Degenerate Case}\\
\hline\hline
&& \multicolumn{5}{c}{} && \multicolumn{2}{c}{Algorithmic}\\
&& \multicolumn{5}{c}{No Algorithmic Subsampling} && \multicolumn{2}{c}{Subsampling}\\
\cline{3-7}\cline{9-10}
&& 0-Way && 1-Way && \multicolumn{4}{c}{2-Way Cluster}\\
\cline{7-10}
$N$ & $M$ & Cluster && Cluster && $p=1$ && $p=1\underline{C}/({NM})$ & $p=2\underline{C}/({NM})$ \\
\hline
40 & 40 &     0.950 && 1.000 && 1.000 && 0.980 & 0.980\\
80 & 80 &     0.940 && 1.000 && 1.000 && 0.971 & 0.975\\
160 & 160 & 0.946 && 1.000 && 1.000 && 0.956 & 0.966\\
320 & 320 & 0.946 && 1.000 && 0.999 && 0.956 & 0.952\\
640 & 640 & 0.953 && 1.000 && 0.999 && 0.946 & 0.952\\
\hline\hline
\end{tabular}
	}
	\caption{95\% coverage frequencies of various inference methods in the nonseparable designs with $N=M=40$, $80$, $160$, $320$ and $640$ based on 2,500 Monte Carlo iterations.}${}$
	\label{tab:simulation_results_separable2_zeroonetwo}
\end{table}

In each of the these two tables, we make the following observations.
First, the 0-Way Cluster method suffers from severe under-coverage across all the sample sizes under the non-degenerate designs. 
Second, the 1-Way Cluster method suffers from under-coverage across all the sample sizes under the non-degenerate designs,
while it in contrast suffers from over-coverage across all the sample sizes under the degenerate designs.

Third, comparisons between the 2-Way Cluster method without algorithmic subsampling and the 2-Way Cluster method with algorithmic subsampling remain the same as those presented in Section \ref{sec:simulation} in the main text.
In particular, we conclude that the 2-Way Cluster with algorithmic subsampling is the only approach that delivers correct coverage across all the designs.
 
\section{Proofs of the Main Results}

\subsection{Proof of Lemma \ref{lemma:consistency}}\label{sec:lemma:consistency}
\begin{proof}
By the definition of $Z_{ij}$, it can be written as $Z_{ij}=\1\{U_{ij} \le p_{NM}\}$ for some i.i.d. $U_{ij}\sim \text{Unif}(0,1)$ independent from the data. 
Define
$\tilde{\mathcal{F}} = \left\{  (u,w) \mapsto f(w):  f \in \mathcal{F} \right\}$
and
$\tilde{\mathcal{G}}_{NM} = \{(u,w)\overset{\tilde g_{{NM}}}{\mapsto} \1(u\le p_{NM})/p_{NM}\}$.
Note that Assumption \ref{a:moment} (ii)--(iii) for $\mathcal{F}$ implies that the same conditions hold with $\tilde{\mathcal{F}}$ in place of $\mathcal F$.
Also, note that for each $(N,M)$, $\tilde{\mathcal{G}}_{NM}$ consists of a single function with itself as an envelope.
Therefore, by Theorem 9.15 in \citet{kosorok2008introduction}, for $\tilde g_{NM} \tilde{\mathcal{F}}=\{\tilde g_{NM}  f : f\in \tilde{\mathcal F}\}$, we have that
\begin{equation*}
\begin{aligned}
&\sup_{Q}N\left(\tilde g_{NM}  \tilde{\mathcal{F}},\left\|\cdot\right\|_{Q,2},\sqrt{2}\epsilon \left\|\tilde g_{NM} F\right\|_{Q,2}\right)
\leq \sup_{Q}N\left(\tilde{\mathcal{F}},\left\|\cdot\right\|_{Q,2},\epsilon \left\|F\right\|_{Q,2}\right) 1
<\infty
\end{aligned}  
\end{equation*}
uniformly over $(N,M)$ for any finite discrete measure $Q$ and $\epsilon\in(0,1]$.
Note that $\tilde g_{NM} \tilde{\mathcal{F}}$ satisfies Assumption \ref{a:moment} (ii)--(iii).
Under Assumptions \ref{a:sampling} and \ref{a:moment} (ii)--(iii), therefore, we can apply Lemma \ref{lemma:d12} (Appendix \ref{sec:useful_lemmas}) to $\tilde g_{NM} \tilde{\mathcal{F}}$, and then apply the Markov inequality to get
$$ \sup_{f\in \mathcal{F}}\left|\frac{1}{NM}\sum_{i=1}^{N}\sum_{j=1}^{M}\frac{Z_{ij}}{p_{NM}}f\left(W_{ij}\right)- \frac{1}{p_{NM}}E\left[Z_{11}f\left(W_{11}\right)\right]\right|\overset{P}{\rightarrow} 0.$$
Since $E\left[Z_{11}f\left(W_{11}\right)\right]=E\left[Z_{11}\right]E\left[f\left(W_{11}\right)\right]=p_{NM}E\left[f\left(W_{11}\right)\right],$
we in turn obtain
$$
\sup_{f\in \mathcal{F}}\left| \frac{1}{L}\sum_{i=1}^{N}\sum_{j=1}^{M}Z_{ij}f\left(W_{ij}\right)- E\left[f\left(W_{11}\right)\right]\right|\overset{P}{\rightarrow} 0.
$$
Finally, Lemma \ref{lemma:inequality} (Appendix \ref{sec:useful_lemmas}) implies that $\hat L/L\overset{P}{\rightarrow} 1$, 
and thus 
$$ 
\sup_{f\in \mathcal{F}}\left|\frac{1}{\hat L}\sum_{i=1}^{N}\sum_{j=1}^{M}Z_{ij}f\left(W_{ij}\right)- E\left[f\left(W_{11}\right)\right]\right|\overset{P}{\rightarrow} 0.
$$
This completes the proof.
\end{proof}

\subsection{Proof of Theorem \ref{theorem:clt}}\label{sec:theorem:clt}
\begin{proof}
Consider the decomposition of $\hat L^{-1}\sum_{i=1}^{N}\sum_{j=1}^{M}Z_{ij}f\left(W_{ij}\right)$ into two terms as
\begin{align}
\frac{1}{\hat L}\sum_{i=1}^{N}\sum_{j=1}^{M}Z_{ij}f\left(W_{ij}\right)=\frac{L}{\hat L}\frac{1}{L}\sum_{i=1}^{N}\sum_{j=1}^{M}Z_{ij}f\left(W_{ij}\right)=\frac{L}{\hat L}\left(A_{NM}+\sqrt{1-p_{NM}}B_{NM}\right),\label{eq:CLT_decomp}
\end{align}
where $A_{NM}$ and $B_{NM}$ are defined by
~\\
\begin{equation*}
    A_{NM}=\frac{1}{NM}\sum_{i=1}^{N}\sum_{j=1}^{M}f\left(W_{ij}\right) \quad \text{and} \quad  B_{NM}=\frac{1}{L}\sum_{i=1}^{N}\sum_{j=1}^{M}
\frac{Z_{ij}-p_{NM}}{\sqrt{1-p_{NM}}}f\left(W_{ij}\right),
\end{equation*}
respectively.

The first step is to get the asymptotic normality for $A_{NM}$. 
Our setup satisfies the first part of Assumption 3 in \citet{davezies2018asymptotic}, since class $\mathcal{F}$ is finite and $E\left[F^2\right]<\infty$. 
Under our Assumptions \ref{a:sampling}, \ref{a:moment} (i), (iii), and \ref{a:scalar}, applying Theorem 3.1 of \citet{davezies2018asymptotic} yields
\begin{equation}\label{eq:a_normal}
\sqrt{\underline{C}}A_{NM}\stackrel{d}{\rightarrow}N\left(0,\lambda_1 E\left[f\left(W_{11}\right)f^T\left(W_{12}\right)\right]+\lambda_2 E\left[f\left(W_{11}\right)f^T\left(W_{21}\right)\right]\right).
\end{equation}

Next, we will obtain the variance-covariance matrix of $B_{NM}$. By the law of total covariance,
\begin{equation*}
\begin{aligned}
\cov\left(B_{NM}, B_{NM}\right)
=&E\left[\cov\left(B_{NM},B_{NM}\middle|\left\{W_{ij}\right\}_{i \in [N],j\in [M]}\right)\right] 
\\
&+\cov\left(E\left[B_{NM}|\{W_{ij}\}_{i \in [N],j\in [M]}\right], E\left[B_{NM}\middle|\{W_{ij}\}_{i \in [N],j\in [M]}\right]\right).
\end{aligned}
\end{equation*}
For the first term, we can write
\begingroup
\allowdisplaybreaks
\begin{align*}
&E\left[\cov\left(B_{NM},B_{NM}\middle|\left\{W_{ij}\right\}_{i \in [N],j\in [M]}\right)\right]\\
&= E\left[\frac{1}{L^2(1-p_{NM})}\sum_{i=1}^{N}\sum_{j=1}^{M}f\left(W_{ij}\right)f^T\left(W_{ij}\right)\var\left(
Z_{ij} \lvert \{W_{ij}\}_{i \in [N],j\in [M]}\right)\right]\\
&=E\left[\frac{1}{L^2(1-p_{NM})}\sum_{i=1}^{N}\sum_{j=1}^{M}f\left(W_{ij}\right)f^T\left(W_{ij}\right)p_{NM}(1-p_{NM})\right]\\
&=\frac{p_{NM}}{L^2}\sum_{i=1}^{N}\sum_{j=1}^{M}E\left[f\left(W_{ij}\right)f^T\left(W_{ij}\right)\right]\\
&=\frac{1}{L}E\left[f\left(W_{ij}\right)f^T\left(W_{ij}\right)\right].
\end{align*}
\endgroup
For the last term, note that
\begin{equation*}
\begin{aligned}
E\left[B_{NM}\middle|\left\{W_{ij}\right\}_{i \in [N],j\in [M]}\right]&=\frac{1}{L\sqrt{1-p_{NM}}}\sum_{i=1}^{N}\sum_{j=1}^{M}f\left(W_{ij}\right)E\left[(Z_{ij}-p_{NM})|\{W_{ij}\}_{i \in [N],j\in [M]}\right]\\
&=\frac{1}{L\sqrt{1-p_{NM}}}\sum_{i=1}^{N}\sum_{j=1}^{M}f\left(W_{ij}\right)E\left[(Z_{ij}|\{W_{ij}\}_{i \in [N],j\in [M]})-p_{NM}\right]\\
&=\frac{1}{L\sqrt{1-p_{NM}}}\sum_{i=1}^{N}\sum_{j=1}^{M}f\left(W_{ij}\right)\left(p_{NM}-p_{NM}\right)=0.
\end{aligned}
\end{equation*}
Therefore,$$\cov\left(E\left[B_{NM}|\{W_{ij}\}_{i \in [N],j\in [M]}\right], E\left[B_{NM}\middle|\{W_{ij}\}_{i \in [N],j\in [M]}\right]\right)=0.$$
It thus follows that
\begin{equation}\label{eq:b_normal}
\begin{aligned}
\cov\left(B_{NM},B_{NM}\right)=\frac{1}{L}E\left[f\left(W_{ij}\right)f^T\left(W_{ij}\right)\right]=\frac{1}{L}E\left[f\left(W_{11}\right)f^T\left(W_{11}\right)\right],
\end{aligned}   
\end{equation}
where the second equality holds by Assumption \ref{a:sampling} (i).

We now show that the term $(A_{NM}+\sqrt{1-p_{NM}} B_{NM})$ is asymptotically normal.
Pick any $q=\left(q_1,...,q_k\right)^T\in \mathbb{R}^k$.
For a given bounded sequence $\left\{a_{ij}\right\}$, define
$$
Y_{NM,L}=\frac{1}{\sqrt{L}}\sum_{i=1}^{N}\sum_{j=1}^{M}
\frac{\left(Z_{ij}-p_{NM}\right)a_{ij}}{\sqrt{1-p_{NM}}} \quad \textrm{and} \quad \alpha^2_{NM}=\frac{1}{NM}\sum_{i=1}^{N}\sum_{j=1}^{M}a^2_{ij}.
$$ 
And 
suppose $f(W_{ij})$ is bounded, by applying Lemma 2 of \citet{Janson1984} with $a_{ij}=q^Tf\left(W_{ij}\right),$ conditionally on $\{W_{ij}\}_{i\in [N],j\in [M]}$, we obtian
$$ E\left(e^{it Y_{NM,L}}\middle|\left\{W_{ij}\right\}_{i\in [N],j\in [M]}\right)-e^{-t^2\alpha^2_{NM}/2}\rightarrow 0.$$
Meanwhile, $\left(NM\right)^{-1}\sum_{i=1}^{N}\sum_{j=1}^{M}f^T\left(W_{ij}\right)qq^Tf\left(W_{ij}\right) \overset{P}\rightarrow E\left[f^T\left(W_{11}\right)qq^Tf\left(W_{11}\right)\right]$, and thus $\alpha^2_{NM} \overset{P}\rightarrow \alpha^2$, where $\alpha^2=E\left[f^T\left(W_{11}\right)qq^Tf\left(W_{11}\right)\right]$, so that the above conditional characteristic function converges to $e^{-t^2\alpha^2/2}$.
Thus, conditionally on  $\left\{W_{ij}\right\}_{i\in [N],j\in [M]}$, we have $\sqrt{L} q^T B_{NM}/\alpha\stackrel{d}{\to}  N(0,1)$. Also note that conditional on $\left\{W_{ij}\right\}_{i\in [N],j\in [M]}$, $A_{NM}$ is deterministic. In addition, we have already shown that $\sqrt{\underline{C}}A_{NM}$ is (unconditionally) asymptotically normal as in \eqref{eq:a_normal}.  Therefore, an application\footnote{We thank a reviewer for suggesting this proof strategy, which simplifies the proof.} of Theorem 2 in \cite{chen2007asymptotic} yields that
\begin{equation}
\begin{aligned}
\sqrt{\underline{C}}\frac{1}{ L}\sum_{i=1}^{N}\sum_{j=1}^{M}Z_{ij}q^T f\left(W_{ij}\right)&=\sqrt{\underline{C}}\frac{1}{NM}\sum_{i=1}^{N}\sum_{j=1}^{M}q^Tf\left(W_{ij}\right)+\frac{\sqrt{\underline{C}}}{\sqrt{L}}\sqrt{1-p_{NM}}\frac{1}{\sqrt{L}}\sum_{i=1}^{N}\sum_{j=1}^{M}
\frac{Z_{ij}-p_{NM}}{\sqrt{1-p_{NM}}}q^Tf\left(W_{ij}\right)\\
&\overset{d}{\rightarrow}N\left(0,q^T\Gamma q\right),
\end{aligned}
\end{equation}
recall that $\Gamma=\Gamma _A+\Lambda \Gamma_B$ with
$
\Gamma_A=\lambda_1 E\left[f\left(W_{11}\right)f^T\left(W_{12}\right)\right]+\lambda_2 E\left[f\left(W_{11}\right)f^T\left(W_{21}\right)\right]$ 
and   
$\Gamma_B=E\left[f\left(W_{11}\right)f^T\left(W_{11}\right)\right].$
 The Cram\'er-Wold device now implies
\begin{equation}
\begin{aligned}
\sqrt{\underline{C}}(A_{NM}+\sqrt{1-p_{NM}} B_{NM})=\sqrt{\underline{C}}\frac{1}{ L}\sum_{i=1}^{N}\sum_{j=1}^{M}Z_{ij} f\left(W_{ij}\right)
\overset{d}{\rightarrow}N\left(0,\Gamma \right).
\end{aligned}
\end{equation}
In case where $f$ is unbounded,  one can approximate $f$ in $L^2$ using a bounded function $f'$ following the argument in Theorem 1 of \citet[pp. 499]{Janson1984} under Lemma \ref{lemma:inequality} and the condition $E[F^2]<\infty$ in the statement of the theorem. The resulting errors in $\sqrt{\underline{C}}A_{NM}$ and $\sqrt{L}B_{NM}$ have variances bounded by $T_{1}E[(f(W_{11})-f'(W_{11}))(f(W_{12})-f'(W_{12}))^T]+T_{2}E[(f(W_{11})-f'(W_{11}))(f(W_{21})-f'(W_{21}))^T]$ and $T_{3}E[(f(W_{11})-f'(W_{11}))(f(W_{11})-f'(W_{11}))^T]$ from \eqref{eq:a_normal} and \eqref{eq:b_normal}, where $T_1,$ $T_2$ and $T_3$ are constants. The result then follows by letting $f'\to f$ with an application of the dominated convergence theorem.

Finally, $\sqrt{\underline{C}} L^{-1}\sum_{i=1}^{N}\sum_{j=1}^{M}Z_{ij}f(W_{ij})$ can be replaced by $\sqrt{\underline{C}}\hat L^{-1}\sum_{i=1}^{N}\sum_{j=1}^{M}Z_{ij}f(W_{ij})$ by virtue of Lemma \ref{lemma:inequality} (Appendix \ref{sec:useful_lemmas}).
\end{proof}

\subsection{Proof of Corollary \ref{pro:OLS}}\label{sec:pro:OLS}

\begin{proof}
Since $E\left[||X_{11}||^4\right]< \infty,$ we have $E\left[X_{r,11}^4\right]< \infty,$ for any coordinate $X_{r,11}$ of $X_{11}.$ 
The condition $E\left[|Y_{11}|^4\right]<\infty$ implies $E\left[u_{11}^4\right]<\infty$.
By Cauchy-Schwarz inequality, we have $E\left[(X_{r,11}u_{11})^2\right]<\infty$ and  $E\left[(X_{r,11}X_{r',11})^2\right]<\infty$ for any $r$ and $r'$, and also $E\left[(||X_{11}u_{11}||)^2\right]< \infty.$

Applying Lemma \ref{lemma:consistency} to the function class $\mathcal{F}_{OLS,1}=\{f(W_{ij})=X_{r,11}X_{r',11}, \, \text{for all} \, r, r'\},$ we have \begin{equation*}
    \frac{1}{\hat L}\sum_{i=1}^{N}\sum_{j=1}^{M}Z_{ij}X_{ij}X_{ij}^T \overset{P}{\rightarrow} E\left[X_{11}X_{11}^T\right].
\end{equation*}

Now, let a vector $\mu$ have the same dimension as $\beta,$ and denote $f_{\mu}(W_{ij})=\mu^TX_{ij}u_{ij}$.
Applying Theorem \ref{theorem:clt} to $\mathcal{F}_{OLS,2}=\{f_{\mu}\left(W_{ij}\right), i \in \{1,...,N\}, j \in \{1,...,M\}\},$ we
obtain \begin{equation*}
    \sqrt{\underline C}\frac{1}{\hat L}\sum_{i=1}^{N}\sum_{j=1}^{M}Z_{ij}f_{\mu}(W_{ij})\overset{d}{\rightarrow}N\left(0,\mu^T\Gamma_{OLS}\mu \right),
\end{equation*}
where $\Gamma_{OLS}=\Gamma_{OLS,1}+\Lambda \Gamma_{OLS,2},$ $\Gamma_{OLS,1}=\lambda_1 E\left[X_{11}u_{11}\left(X_{12}u_{12}\right)^T\right]+\lambda_2 E\left[X_{11}u_{11}\left(X_{21}u_{21}\right)^T\right]$
and   
$\Gamma_{OLS,2}=E\left[X_{11}u_{11}\left(X_{11}u_{11}\right)^T\right].$
Cram\'er-Wold device thus yields
\begin{equation*}
    \sqrt{\underline C}\frac{1}{\hat L}\sum_{i=1}^{N}\sum_{j=1}^{M}Z_{ij}X_{ij}u_{ij}\overset{d}{\rightarrow}N\left(0,\Gamma_{OLS} \right),
\end{equation*}

Finally, applying Slutsky's lemma yields
\begin{equation*}
    \sqrt{\underline C}\left(\hat \beta-\beta\right)\overset{d}{\rightarrow}N\left(0,V\right),
\end{equation*}
where $V=J^{-1}\Gamma_{OLS} J^{-1}$ and $J=E\left[X_{11}X_{11}^T\right].$
\end{proof}
\section{Useful Lemmas}\label{sec:useful_lemmas}

In this appendix section, we state auxiliary lemmas that are used to prove our main results. 
Each of these results is either coming directly from the existing literature or is the existing result with minor modifications.
For the latter case, we provide a proof.

\begin{lemma}\label{lemma:pointwise_measurable}
Let $\Theta$ be a compact subset of $\mathbb{R}^k$ and let $\mathcal{F}=\left\{f\left(\cdot,\theta\right): \theta \in \Theta\right\}$ be a class of real-valued functions indexed by $\theta$ such that $f(w,\cdot)$ is continuous for all $w\in \supp(W_{ij})$. Then, $\mathcal{F}$ is a pointwise measurable class of functions. 
\end{lemma}

\begin{proof}
The proof is immediate and well-known. We provide the proof for completeness.
Let $\mathcal{S}=\left\{f\left(\cdot,\theta\right): \theta \in \Theta \cap \mathbb{Q}^k \right\}$, where $\mathbb{Q}$ is the rationals.  
Therefore, by the denseness of $\mathbb{Q}^k$, for each $w\in \supp(W_{ij})$, we can find $(\theta_m)\subset \Theta \cap \mathbb{Q}^k$, $\theta_m\to \theta$ as $m\to \infty$ and then the continuity implies $f\left(w,\theta_m\right) \rightarrow f\left(w,\theta\right)$, which coincides with the definition of pointwise measurability.
\end{proof}
\noindent The next lemma follows immediately from \citet[Lemma 2.2.9]{van1996weak}.
\begin{lemma}[Bernstein's Inequality for Bernoulli r.v.'s]\label{lemma:inequality}
\noindent 
For each $p\in(0,1]$, it holds that
$$P\left(|\hat L/L -1|>\sqrt{2t/L}+2t/(3L)\right)\leq 2e^{-t},$$
for every $t>0$. 
\end{lemma}

\noindent Lemmas \ref{lemma:d12} and \ref{lemma:d11} bellow follow closely from Theorem 3.4(i) in \citet{davezies2019empirical} and Lemma D.12 in \citet{davezies2018asymptotic}, respectively. 


\begin{lemma}[Glivenko-Cantelli for two-way clustered random variables]\label{lemma:d12}
Let $(\mathcal{F}_{NM})$ be a sequence of classes of functions that satisfies Assumption \ref{a:moment} (iii) and such that each $\mathcal{F}_{NM}$ admits an envelop $F_{NM}$ with $E\left[F_{NM}\left(W_{11}\right)\right]\le \overline M<\infty,$ $\sup_Q\log N\left(\mathcal{F}_{NM}, \left\|\cdot\right\|_{Q,2}, \epsilon\left\|F_{NM}\right\|_{Q,2}\right)<\infty$ for any finite discrete measure $Q$, $\epsilon\in (0,1],$ then under Assumption \ref{a:sampling}, we have $$E\left[\sup_{f\in\mathcal{F}_{NM}}\left|\frac{1}{NM}\sum_{i=1}^{N}\sum_{j=1}^{M}f\left(W_{ij}\right)-E\left[f\left(W_{11}\right)\right]\right|\right]=o(1).$$
\end{lemma}
\begin{proof}
	The result is a minor modification of the proof of Theorem 3.4 (i) in \citet{davezies2019empirical} with the standard Glivenko-Cantelli theorem modified for function classes changing with the sample size.
	Denote $\mathbb P_{NM}=(NM)^{-1}\sumij \delta_{X_{ij}}$, where $\delta_{x}$ is the Dirac measure at $x$.
	 Following their symmetrization argument (which is nonasymptotic and independent of the function class) in the proof of Theorem 3.4 (i) in \citet{davezies2019empirical}, for each $K>0$ and $\epsilon>0$, denote $\mathcal F_{NM,K}=\mathcal F_{NM} \1\{F_{NM}>K\}$, then one has
	\begin{align*}
	&E\left[\sup_{f\in\mathcal{F}_{NM}}\left|\frac{1}{NM}\sum_{i=1}^{N}\sum_{j=1}^{M}f\left(W_{ij}\right)-E\left[f\left(W_{11}\right)\right]\right|\right]\\\lesssim& E[F_{NM} \1 \{F_{NM}>K\}]+ E\left[\epsilon +\frac{K}{\sqrt{NM}}\sqrt{\log N(\mathcal F_{NM,K},\|\cdot\|_{\mathbb P_{NM},1} ,\epsilon)}\right].
	\end{align*}
	The first term on the right hand side is bounded by $\overline M$. To deal with the second term, by Jensen's inequality, it holds that $\|f-f'\|_{\mathbb P_{NM},1}\le \|f-f'\|_{\mathbb P_{NM},2}$. Thus the smallest $\epsilon$-net for $(\mathcal F_{NM},\|\cdot\|_{\mathbb P_{NM},2} )$ is an $\epsilon$-net for $(\mathcal F_{NM},\|\cdot\|_{\mathbb P_{NM},1} )$. Thus we have 
	 $N(\mathcal F_{NM},\|\cdot\|_{\mathbb P_{NM},1} ,\epsilon)\le N(\mathcal F_{NM},\|\cdot\|_{\mathbb P_{NM},2} ,\epsilon)$. The condition $\sup_Q\log N\left(\mathcal{F}_{NM}, \left\|\cdot\right\|_{Q,2}, \epsilon\left\|F_{NM}\right\|_{Q,2}\right)<\infty$ for all $\epsilon\in(0,1]$ implies $N(\mathcal F_{NM},\|\cdot\|_{\mathbb P_{NM},2} ,\epsilon)<\infty$ for all $\epsilon>0$. Finally, observe that $E[\|F\|_{\mathbb P_{NM},1}]=E[F_{NM}]<M$. This concludes the proof. 
\end{proof}

\begin{lemma}[Lemma D.11 in \citet{davezies2018asymptotic} for sequences]\label{lemma:d11}
Let $(\mathcal{F}_{NM})$ and $(\mathcal{G}_{NM})$ be two pointwise measurable classes of functions. Suppose that each $\mathcal{F}_{NM}$ admits an envelope $F_{NM}$ with $E\left[F_{NM}\left(W_{11}\right)^2\right]<\infty$ and 
$$\int_{0}^{1}\sup_Q\sqrt{\log N\left(\mathcal{F}_{NM}, \left\|\cdot\right\|_{Q,2},\epsilon\left\|F_{NM}\right\|_{Q,2}\right)}d\epsilon\le \overline M<\infty,$$ where $Q$ is taken over the set of all finite discrete measures and $\epsilon\in (0,1].$ Similarly, $(\mathcal{G}_{NM})$ admits a sequence of envelop functions $(G_{NM})$ with $E\left[G_{NM}\left(W_{11}\right)^2\right]<\infty$ and $$\int_{0}^{1}\sup_Q\sqrt{\log N\left(\mathcal{G}_{NM}, \left\|\cdot\right\|_{Q,2},\epsilon\left\|G_{NM}\right\|_{Q,2}\right)}d\epsilon\le \overline M<\infty.$$ Then, under Assumptions \ref{a:sampling} and \ref{a:scalar}, $$\lim_{\underline C\rightarrow \infty}E\left[\sup_{\mathcal{F}_{NM}\times \mathcal{G}_{NM}}\left|\frac{\underline{C}}{\left(NM\right)^2}\sum_{i=1}^{N}\sum_{1\leq j, j'\leq M}f\left(W_{ij}\right)g\left(W_{ij'}\right)-\lambda_1E\left[f\left(W_{11}\right)g\left(W_{12}\right)\right]\right|\right]=0.$$
\end{lemma}
\begin{proof}
The proof follows the same steps of the proof of Lemma D.11 of \citet{davezies2018asymptotic} with the modification of $\mathcal F_{NM}$, $F_{NM}$, $\mathcal G_{NM}$, and $G_{NM}$ in place of $\mathcal F$, $F$, $\mathcal G$, and $G$, respectively. Notice that their symmetrization arguments and Lemma D.4 are non-asymptotic and thus are not affected by such modification. The detail is omitted.
\end{proof}

\section{Proofs for the Application to the GMM} 

Define the class $\mathcal{G}=\left\{g_{r}\left(\cdot,\theta\right): \theta \in \Theta \, \, \textrm{and}\,\, r\in \{1,...,m\}\right\}$ of functions indexed by $r$ and $\theta$ and the class $\mathcal{G}'=\bigl\{\partial g_r(\cdot,\theta)/\partial \theta_l:\theta \in \Theta, r \in \{1,...,m\}, l \in \{1,...,k\}\bigl\}$ of functions indexed by $r$, $l$ and $\theta.$

\subsection{Proof of Lemma \ref{theorem:consistency_gmm}}\label{sec:theorem:consistency_gmm}
\begin{proof}
We verify the conditions of Theorem 2.1 in \citet{Newey1994}, where the population criterion is $Q_0(\theta)=$ $-E\left[g\left(W_{ij},\theta\right)\right]^TVE\left[g\left(W_{ij},\theta\right)\right]$. 
Their Condition 2.1 (i), $Q_0\left(\theta\right)$ is uniquely maximized at $\theta^0$, holds by Lemma 2.3 in \citet{Newey1994} under Assumption \ref{a:GMM} (i).
Condition 2.1 (ii) holds by Assumption \ref{a:GMM} (ii). 
Condition 2.1 (iii) that $Q_0$ is continuous at $\theta$ follows from Assumption \ref{a:GMM} (iii) (a). 
Under Assumption \ref{a:GMM} (ii), (iii) (a), (iv), by Example 19.7 in \citet{van2000asymptotic} and Lemma 9.18 in \citet{kosorok2008introduction}, we know that the class 
$\mathcal{G}_r=\left\{g_r\left(\cdot,\theta\right),\theta \in \Theta\right\}$, for $r\in \left\{1,...,m\right\},$ 
has an envelope $\left|g_r(W_{ij},\theta^0)\right|+DM < \infty,$ where $D$ is a diameter of a set containing $\Theta$. 
Thus, for any finite discrete measure $Q$ and $\epsilon\in (0,1]$, $N\left(\mathcal{G}_r,\left\|\cdot\right\|_{Q,2},\epsilon\right)\leq \left(1+4DM/\epsilon\right)^k.$ 
Since $\mathcal{G}=\bigcup_{r=1}^{m}\mathcal{G}_r,$ we obtain $N\left(\mathcal{G},\left\|\cdot\right\|_{Q,2},\epsilon\right)\leq m\left(1+4DM/\epsilon\right)^k<\infty$, which implies that the class $\mathcal{G}$ satisfies Assumption \ref{a:moment} (ii). $\mathcal{G}$ is a pointwise measurable class of functions since $\mathcal{G}_r$, for $r\in \left\{1,...,m\right\},$ is a pointwise measurable class of functions by Lemma \ref{lemma:pointwise_measurable} under Assumption \ref{a:GMM} (ii), (iii) (a) and $\mathcal{G}=\bigcup_{r=1}^{m}\mathcal{G}_r.$
Thus, with Assumption \ref{a:sampling}, by applying Lemma \ref{lemma:consistency}, we have 
 $\sup_{\theta \in \Theta}\left\|\hat g_{NM}\left( \theta\right)-E\left[g\left(W_{ij},\theta\right)\right]\right\| \overset{P}{\rightarrow} 0.$ 
By the triangle and Cauchy-Schwartz inequalities, we obtain
\begin{equation*}
\begin{aligned}
&\left|\hat Q_{NM}\left(\theta\right)-Q_0\left(\theta\right)\right|\\
&\leq \left|\left(\hat g_{NM}\left(\theta \right)-E\left[g\left(W_{ij},\theta\right)\right]\right)^T\hat V\left(\hat g_{NM}\left( \theta\right)-E\left[g\left(W_{ij},\theta\right)\right]\right)\right|\\
&+\left|E\left[g\left(W_{ij},\theta\right)\right]^T\left(\hat V+\hat V^T\right)\left(\hat g_{NM}\left( \theta\right)-E\left[g\left(W_{ij},\theta\right)\right]\right)\right|+\left|E\left[g\left(W_{ij},\theta\right)\right]^T\left(\hat V-V\right)E\left[g\left(W_{ij},\theta\right)\right]\right|\\
&\leq \left\|\hat g_{NM}\left( \theta\right)-E\left[g\left(W_{ij},\theta\right)\right]\right\|^2\left\|\hat V\right\|+2\left\|E\left[g\left(W_{ij},\theta\right)\right]\right\|\left\|\hat g_{NM}\left( \theta\right)-E\left[g\left(W_{ij},\theta\right)\right]\right\| \left\|\hat V\right\|\\ &+\left\|E\left[g\left(W_{ij},\theta\right)\right]\right\|^2\left\|\hat V-V\right\|.
\end{aligned}
\end{equation*}
Thus, $\sup_{\theta \in \Theta}\left|\hat Q_{NM}\left(\theta\right)-Q_0\left(\theta\right)\right|\overset{P}{\rightarrow} 0$ so that condition 2.1 (iv) is satisfied. 
Applying Theorem 2.1 in \citet{Newey1994}, we therefore obtain $\hat \theta \overset{P}{\rightarrow} \theta^0.$
\end{proof}

\subsection{Proof of Theorem \ref{theorem:asymptotic_normality_gmm}}\label{sec:theorem:asymptotic_normality_gmm}
\begin{proof}
Under Assumption \ref{a:GMM} (ii), (iii) (a), the first-order condition requires that $2\hat G_{NM}\left(\hat \theta\right)^T\hat V\hat g_{NM}\left(\hat \theta\right) = 0$ holds with probability approaching one, where $\hat G_{NM}\left(\theta\right)=\nabla_{\theta}\hat g_{NM}\left( \theta\right).$ Expanding $\hat g_{NM}\left(\hat \theta\right)$ around $\theta^0$ and multiplying by $\sqrt{\underline C}$, we have 
$$
\sqrt{\underline C}\left(\hat \theta-\theta^0\right)=-\left[\hat G_{NM}\left(\hat \theta\right)^T\hat V\hat G_{NM}\left(\bar \theta\right)\right]^{-1}\hat G_{NM}\left(\hat \theta\right)^T\hat V \sqrt{\underline{C}}\hat g_{NM}\left(\theta^0\right),
$$ 
where $\bar \theta$ is the mean value implied by the mean value theorem for each coordinate. Under Assumption \ref{a:GMM} (ii), (iii) (b), similar lines of argument to those in the proof of Lemma \ref{theorem:consistency_gmm} yield $N\left(\mathcal{G'},\left\|\cdot\right\|_{Q,2},\epsilon\right)<\infty$ for any finite discrete measure $Q$ and $\epsilon\in (0,1]$.
$\mathcal{G}'$ is a pointwise measurable class of functions by Lemma \ref{lemma:pointwise_measurable} under Assumption \ref{a:GMM} (ii), (iii) (b). 
With Assumptions \ref{a:sampling} and \ref{a:GMM} (vi), Lemma \ref{lemma:consistency} thus yields 
$$
\sup_{\theta \in \Theta}\left|\frac{1}{\hat L}\sum_{i=1}^{N}\sum_{j=1}^{M}Z_{ij}\frac{\partial g_r\left(W_{ij}, \theta\right)}{\partial \theta_l}- E\left[\frac{\partial g_r\left(W_{11}, \theta\right)}{\partial \theta_l}\right]\right| \overset{P}{\rightarrow}0
$$ 
for each $r$ and $l.$ Since there are only finite numbers of $r$ and $l,$ it follows that $\hat G_{NM}\left(\hat \theta\right) - G\left(\hat \theta\right) \overset{P}{\rightarrow} 0$ and $\hat G_{NM}\left(\bar \theta\right) - G\left(\bar \theta\right) \overset{P}{\rightarrow} 0,$ where $G(\theta)=E[\nabla_{\theta}g(W_{11},\theta)].$ 
Also, since the conditions of Lemma \ref{theorem:consistency_gmm} are satisfied, we have $\bar \theta \overset{P}{\rightarrow} \theta^0$ and $\hat \theta \overset{P}{\rightarrow} \theta^0$. We thus obtain $ G\left(\hat \theta\right) - G \overset{P}{\rightarrow} 0$ and $G\left(\bar \theta\right) - G \overset{P}{\rightarrow} 0$ by the continuous mapping theorem under Assumption \ref{a:GMM} (iii) (b).
Combining the above results yields $\hat G_{NM}\left(\hat \theta\right)\overset{P}{\rightarrow} G$ and $\hat G_{NM}\left(\bar \theta\right)\overset{P}{\rightarrow} G.$ 
Therefore, $\left[\hat G_{NM}\left(\hat \theta\right)^T\hat V\hat G_{NM}\left(\bar \theta \right)\right]^{-1}\hat G_{NM}\left(\hat \theta\right)^T\hat V  \overset{P}{\rightarrow}\left(G^TVG\right)^{-1}G^TV$ follows by an application of the continuous mapping theorem under Assumption \ref{a:GMM} (v).  
Now, notice that finite function class $\{g_1\left(\cdot,\theta^0\right),..., g_m\left(\cdot,\theta^0\right)\}$ is pointwise measurable since $\mathcal{G}$ is a pointwise measurable class of functions following
Lemma \ref{theorem:consistency_gmm} and $E\left[g\left(W_{ij},\theta^0\right)\right]=0.$ With $E\left[g_{\sup}(W_{ij})^2\right]<\infty$ under assumption \ref{a:GMM} (vii), by applying Theorem \ref{theorem:clt} under Assumptions \ref{a:sampling} and \ref{a:scalar}, we obtain $\sqrt{\underline{C}}\hat g_{NM}\left(\theta^0\right)\overset{d}{\rightarrow} N\left(0,\Omega\right)$, where $\Omega=\Gamma _1+\Lambda \Gamma_2$, with $\Gamma _1= \lambda_1E\left[g\left(W_{11},\theta^0\right)g^T\left(W_{12},\theta^0\right)\right]+\lambda_2E\left[g\left(W_{11},\theta^0\right)g^T\left(W_{21},\theta^0\right)\right]$ and $\Gamma_2= E\left[g\left(W_{11},\theta^0\right)g^T\left(W_{11},\theta^0\right)\right].$ 
The Slutsky's theorem then implies
$$
\sqrt{\underline{C}}\left(\hat \theta-\theta^0\right)\overset{d}{\rightarrow} N\left(0,\left(G^TVG\right)^{-1}G^TV\Omega VG\left(G^TVG\right)^{-1}\right),
$$ 
which concludes the proof.
\end{proof}

\subsection{Proof of Theorem \ref{theorem:consistency_gmm_variance}}\label{sec:theorem:consistency_gmm_variance}
\begin{proof}
First, we want to establish $\widetilde G\overset{P}{\rightarrow} G$ via
$
\left\|\widetilde G-G\right\|\leq \left\|\widetilde G-G\left(\hat \theta\right)\right\|+\left\|G\left(\hat \theta\right)-G\right\|,
$
where $G\left(\theta\right)=E\left[\nabla_\theta g\left(W_{11},\theta\right)\right].$ 
Since the conditions of Lemma \ref{theorem:consistency_gmm} are satisfied, it holds that $\hat \theta \overset{P}{\rightarrow}\theta^0.$ Under Assumption \ref{a:GMM:variance} (i), we obtain $\left\|G\left(\hat \theta\right)-G\right\|\overset{P}{\rightarrow} 0$ by the continuous mapping theorem. 
Note that $\mathcal{G}'$ is pointwise measurable and $N\left(\mathcal{G'},\left\|\cdot\right\|_{Q,2},\epsilon\right)<\infty$ for any finite discrete measure $Q,$ $\epsilon\in (0,1]$ by the proof of Theorem \ref{theorem:asymptotic_normality_gmm}. 
Under Assumptions \ref{a:sampling} and \ref{a:GMM} (vi), by applying Lemma \ref{lemma:consistency}, we thus obtain 
$$
\sup_{\theta \in \Theta}\left|\frac{1}{\hat L}\sum_{i=1}^{N}\sum_{j=1}^{M}\frac{\partial g_r\left(W_{ij}, \theta\right)Z_{ij}}{\partial \theta_l}- E\left[\frac{\partial g_r\left(W_{11}, \theta\right)}{\partial \theta_l}\right]\right| \overset{P}{\rightarrow}0
$$ 
for each $r$ and $l.$ 
Since there are only finite numbers of $l$ and $r$, we get 
$$
\sup_{\theta \in \Theta}\left\|\frac{1}{\hat L} \sum_{i=1}^{N}\sum_{j=1}^{M}Z_{ij}\nabla_{\theta}g\left(W_{ij},\theta\right)-E\left[\nabla_{\theta }g\left(W_{11},\theta\right)\right]\right\|\overset{P}{\rightarrow}0.
$$ 
It then follows that $\tilde G\overset{P}{\rightarrow}G(\hat \theta).$ 
Combining the above arguments, we establish $\tilde G\overset{P}{\rightarrow} G.$ 

We will next verify $\tilde\Gamma_2\overset{P}{\rightarrow} \Gamma_2.$
Define a new class $\mathcal{G}_{\text{sub}}=\left\{(w,z) \mapsto z g_r\left(w, \theta\right), \theta \in \Theta, r \in \left\{1,...,m\right\}\right\}$. 
For any finite discrete measure $Q$ and $\epsilon\in (0,1],$ we have $\sup_{Q}N\left(\mathcal{G},\left\|\cdot\right\|_{Q,2},\epsilon\left\|g_{\sup}\right\|_{Q,2}\right)< \infty$ by the proof of Lemma \ref{theorem:consistency_gmm}. 
By Theorem 9.15 in \citet{kosorok2008introduction}, therefore,
     \begin{equation*}
     \begin{aligned}
      &\sup_{Q}N\left(\mathcal{G}_{\text{sub}}\mathcal{G}_{\text{sub}},\left\|\cdot\right\|_{Q,2},\sqrt{2}\epsilon\left\|g_{\sup}^2\right\|_{Q,2}\right)\\
      &\le \sup_{Q}N\left(\mathcal{G}_{\text{sub}},\left\|\cdot\right\|_{Q,2},\epsilon\left\|g_{\sup}\right\|_{Q,2}\right)\sup_{Q}N\left(\mathcal{G}_{\text{sub}},\left\|\cdot\right\|_{Q,2},\epsilon\left\|g_{\sup}\right\|_{Q,2}\right)<\infty,
     \end{aligned}
 \end{equation*} 
for any finite discrete measure $Q$ and $\epsilon\in (0,1]$, where $\mathcal{G}_{\text{sub}}\mathcal{G}_{\text{sub}}$ is defined as the pointwise product. 
Note that $\mathcal{G}_{\text{sub}}\mathcal{G}_{\text{sub}}$ is a pointwise measurable class of functions since $\mathcal{G}_{\text{sub}}$ is a pointwise measurable class of functions by the arguments in the proof of Lemma \ref{theorem:consistency_gmm}.  This implies that 
with $E[g_{\sup}(W_{ij})^2]<\infty$ and Assumption \ref{a:sampling}, by applying Lemma \ref{lemma:d12}, we thus obtain 
$$
E\left[\sup_{\theta \in \Theta}\left\|\frac{1}{NM}\sum_{i=1}^N\sum_{j=1}^M g\left(W_{ij}, \theta\right)Z_{ij}\left(g\left(W_{ij}, \theta\right)Z_{ij}\right)^T-E\left[g\left(W_{11},  \theta\right)Z_{11}\left(g\left(W_{11},  \theta\right)Z_{11}\right)^T\right]\right\|\right]=o(1).
$$
As we can write 
\begin{equation*}
    \begin{aligned}
     \frac{1}{NM}\sum_{i=1}^N\sum_{j=1}^M g\left(W_{ij}, \theta\right)Z_{ij}\left(g\left(W_{ij}, \theta\right)Z_{ij}\right)^T&=\frac{L}{NM}\frac{1}{L}\sum_{i=1}^N\sum_{j=1}^M g\left(W_{ij}, \theta\right)Z_{ij}\left(g\left(W_{ij}, \theta\right)Z_{ij}\right)^T\\
     &=p\frac{1}{L}\sum_{i=1}^N\sum_{j=1}^Mg\left(W_{ij}, \theta\right)Z_{ij}\left(g\left(W_{ij}, \theta\right)Z_{ij}\right)^T
    \end{aligned}
\end{equation*}
and 
$$
E\left[g\left(W_{11},  \theta\right)Z_{11}\left(g\left(W_{11},  \theta\right)Z_{11}\right)^T\right]=E\left[Z_{11}^2\right]E\left[g\left(W_{11},  \theta\right)g^T\left(W_{11}, \theta\right)\right]=pE\left[g\left(W_{11},  \theta\right)g^T\left(W_{11}, \theta\right)\right].
$$
Therefore, by Markov's inequality, it follows that 
$$
\sup_{\theta \in \Theta}\left\|\frac{1}{L}\sum_{i=1}^N\sum_{j=1}^M g\left(W_{ij}, \theta\right)Z_{ij}\left(g\left(W_{ij}, \theta\right)Z_{ij}\right)^T-E\left[g\left(W_{11},  \theta\right)g^T\left(W_{11},  \theta\right)\right]\right\|\overset{P}{\rightarrow}0.
$$ 
In addition, 
$$
\frac{1}{L}\sum_{i=1}^N\sum_{j=1}^M Z_{ij}g\left(W_{ij}, \theta\right)g\left(W_{ij}, \theta\right)=\frac{1}{L}\sum_{i=1}^N\sum_{j=1}^M g\left(W_{ij}, \theta\right)Z_{ij}\left(g\left(W_{ij}, \theta\right)Z_{ij}\right)^T.
$$
Thus, Lemma  \ref{lemma:inequality} yields $\tilde \Gamma_2\overset{P}{\rightarrow}\Gamma_2(\hat \theta),$ where $\Gamma_2(\theta)=E\left[g\left(W_{11}, \theta\right)g^T\left(W_{11}, \theta\right)\right].$ Meanwhile, $\hat \theta\overset{P}{\rightarrow}\theta^0$ and we have $\Gamma_2(\hat \theta)\overset{P}{\rightarrow}\Gamma_2$ by Assumption \ref{a:GMM:variance} (iii). 
Therefore, $\tilde\Gamma_2\overset{P}{\rightarrow} \Gamma_2$ follows.

Finally, we establish $\widetilde \Gamma_1\overset{P}{\rightarrow}\Gamma_1$.
Note that $\mathcal{G}_{\text{sub}}$ is a pointwise measurable class of functions and that $\sup_{Q}N\left(\mathcal{G}_{\text{sub}},\left\|\cdot\right\|_{Q,2},\epsilon\left\|g_{\sup}\right\|_{Q,2}\right)<\infty$ for any finite discrete measure $Q$ and $\epsilon\in(0,1]$. 
With $E[g_{\sup}(W_{ij})^2]<\infty$ and Assumptions \ref{a:sampling}, \ref{a:scalar}, Lemma \ref{lemma:d11} yields
$$
\lim_{\underline{C}\rightarrow \infty}E\left[\sup_{\theta \in \Theta}\left\|\frac{\underline{C}}{(NM)^2}\sum_{i=1}^{N}\sum_{1\leq j,j'\leq M}g\left(W_{ij},\theta\right)Z_{ij}\left(g\left(W_{ij'},\theta\right)Z_{ij'}\right)^T-E\left[\lambda_1 g\left(W_{11},\theta\right)Z_{11}\left(g\left(W_{12},\theta\right)Z_{12}\right)^T\right]\right\|\right]=0.
$$
As we can write 
\begin{equation*}
    \begin{aligned}
     &\frac{\underline{C}}{(NM)^2}\sum_{i=1}^{N}\sum_{1\leq j,j'\leq M}g\left(W_{ij},\theta\right)Z_{ij}\left(g\left(W_{ij'},\theta\right)Z_{ij'}\right)^T\\
		&=\frac{L^2}{(NM)^2}\frac{\underline{C}}{L^2}\sum_{i=1}^{N}\sum_{1\leq j,j'\leq M}g\left(W_{ij},\theta\right)Z_{ij}\left(g\left(W_{ij'},\theta\right)Z_{ij'}\right)^T\\
     &=p^2\frac{\underline{C}}{L^2}\sum_{i=1}^{N}\sum_{1\leq j,j'\leq M}g\left(W_{ij},\theta\right)Z_{ij}\left(g\left(W_{ij'},\theta\right)Z_{ij'}\right)^T,\\
    \end{aligned}
\end{equation*}
and 
\begin{align*}
E\left[\lambda_1 g\left(W_{11},\theta\right)Z_{11}\left(g\left(W_{12},\theta\right)Z_{12}\right)^T\right]
&=
E\left[Z_{11}Z_{12}\right]E\left[\lambda_1 g\left(W_{11},\theta\right)g^T\left(W_{12},\theta\right)\right]\\
&=
p^2E\left[\lambda_1 g\left(W_{11},\theta\right)g^T\left(W_{12},\theta\right)\right].
\end{align*}
Therefore, by Markov inequality, it follows that 
$$
\sup_{\theta \in \Theta}\left\|\frac{\underline{C}}{L^2}\sum_{i=1}^{N}\sum_{1\leq j,j'\leq M}g\left(W_{ij},\theta\right)Z_{ij}\left(g\left(W_{ij'},\theta\right)Z_{ij'}\right)^T-\lambda_1E\left[g\left(W_{11},\theta\right)g^T\left(W_{12},\theta\right)\right]\right\|\overset{P}{\rightarrow}0
$$
as $\underline{C}\rightarrow \infty$.
In addition, a symmetric argument also shows that
$$
\sup_{\theta \in \Theta}\left\|\frac{\underline{C}}{L^2}\sum_{1\leq i,i'\leq N}\sum_{j=1}^Mg\left(W_{ij},\theta\right)Z_{ij}\left(g\left(W_{i'j},\theta\right)Z_{i'j}\right)^T-\lambda_2E\left[g\left(W_{11},\theta\right)g^T\left(W_{21},\theta\right)\right]\right\|\overset{P}{\rightarrow}0
$$ 
as $\underline C\rightarrow \infty$.
Also note that
$$
\frac{\underline{C}}{L^2}\sum_{i=1}^{N}\sum_{1\leq j,j'\leq M}Z_{ij}Z_{ij'}g\left(W_{ij},\theta\right)g^T\left(W_{ij'},\theta\right)=\frac{\underline{C}}{L^2}\sum_{i=1}^{N}\sum_{1\leq j,j'\leq M}g\left(W_{ij},\theta\right)Z_{ij}\left(g\left(W_{ij'},\theta\right)Z_{ij'}\right)^T
$$
and 
$$
\frac{\underline{C}}{L^2}\sum_{1\leq i,i'\leq N}\sum_{j=1}^MZ_{ij}Z_{i'j}g\left(W_{ij},\theta\right)g^T\left(W_{i'j},\theta\right)=\frac{\underline{C}}{L^2}\sum_{1\leq i,i'\leq N}\sum_{j=1}^Mg\left(W_{ij},\theta\right)Z_{ij}\left(g\left(W_{i'j},\theta\right)Z_{i'j}\right)^T
$$
hold.
Therefore, Lemma \ref{lemma:inequality} yields $\widetilde \Gamma_1\overset{P}{\rightarrow}\Gamma_1(\hat \theta)$ where
$$
\Gamma_1(\theta)=\lambda_1E\left[g\left(W_{11},\theta\right)g^T\left(W_{12},\theta\right)\right]+\lambda_2E\left[g\left(W_{11},\theta\right)g^T\left(W_{21},\theta\right)\right].
$$ 
Meanwhile, since $\hat \theta\overset{P}{\rightarrow} \theta^0,$ we get $\Gamma_1(\hat \theta)\overset{P}{\rightarrow}\Gamma_1$ by Assumption \ref{a:GMM:variance} (ii). We thus obtain $\widetilde \Gamma_1\overset{P}{\rightarrow}\Gamma_1$.

Combining the above results yields $\widetilde \Omega\overset{P}{\rightarrow} \Omega$ by continuous mapping theorem. 
By another application of the continuous mapping theorem and the continuity of matrix inversion under Assumption \ref{a:GMM} (v), we get $\left(\widetilde G^T\hat V\widetilde G\right)^{-1}\overset{P}{\rightarrow}\left(G^TVG\right)^{-1}.$ 
Therefore, $\left(\tilde G^T\hat V\tilde G\right)^{-1}\tilde G^T\hat V\tilde \Omega \hat V\tilde G\left(\tilde G^T\hat V\tilde G\right)^{-1}$ is consistent for $\left(G^TVG\right)^{-1}G^TV\Omega VG\left(G^TVG\right)^{-1}$ by the continuous mapping theorem.
\end{proof}

\section{Proofs for the Application to the M-Estimation} 

For convenience, define the class $\mathcal{Q}''=\left\{\partial^2q\left(\cdot,\theta\right)/\partial \theta_r\partial \theta_l: \theta\in \Theta \; \textrm{and} \; r,l \in \left\{1,...,k\right\}\right\}$ of functions indexed by $\theta$, $r$, and $l.$ 

\subsection{Proof of Lemma \ref{theorem:consistency_M}}\label{sec:theorem:consistency_M}
\begin{proof} We verify the conditions of Theorem 2.1 in \citet{Newey1994}. 
Condition 2.1 (i) that $Q_0\left(\theta\right)$ is uniquely maximized at $\theta^0$ holds follows from the second part in Assumption \ref{a:M} (i).
Condition 2.1 (ii) holds by the first part in Assumption \ref{a:M} (i). 
Condition 2.1 (iii) that $Q_0\left(\theta\right)$ is continuous at $\theta$ follows from Assumption \ref{a:M} (ii) (a). 
Under Assumption \ref{a:M} (i), (ii) (a), (iii), by Example 19.7 in \citet{van2000asymptotic} and Lemma 9.18 in \citet{kosorok2008introduction}, the class $\mathcal{Q}$ has an envelope $\left|q\left(W_{ij},\theta^0\right)\right|+DM< \infty,$ where $D$ is the diameter of a set containing $\Theta$. 
Thus, for any finite discrete measure $Q$ and $\epsilon\in (0,1]$, $N\left(\mathcal{Q},\left\|\cdot\right\|_{Q,2},\epsilon\right)\leq \left(1+4DM/\epsilon\right)^k$, which implies that the class $\mathcal{Q}$ satisfies Assumption \ref{a:moment} (ii). 
$\mathcal{Q}$ is a pointwise measurable class of functions by Lemma \ref{lemma:pointwise_measurable} under Assumption \ref{a:M} (i), (ii) (a). 
Thus, with Assumption \ref{a:sampling}, by applying Lemma \ref{lemma:consistency}, we have $\sup_{\theta \in \Theta}\left|\hat Q_{NM}\left(\theta\right)-Q_0\left(\theta\right)\right|\overset{P}{\rightarrow}0$, so that condition 2.1 (iv) is satisfied. 
Applying Theorem 2.1 in \citet{Newey1994}, we obtain $\hat \theta \overset{P}{\rightarrow}\theta^0$.
\end{proof}

\subsection{Proof of Theorem \ref{theorem:asymptotic_normality_M}}\label{sec:theorem:asymptotic_normality_M}
\begin{proof}
We verify the conditions of Theorem 3.1 in \citet{Newey1994}. Conditions 3.1 (i), (ii) and (v) follow from the Assumption \ref{a:M} (i), (ii) (c), (v). 
Note that the finite function class $\left\{\partial q(\cdot,\theta^0)/\partial \theta_1,...,\partial q(\cdot,\theta^0)/\partial \theta_k\right\}$ is pointwise measurable since $\mathcal{Q}'$ is by Lemma \ref{lemma:pointwise_measurable} under Assumption \ref{a:M} (i), (ii) (b). With $E[\dot q_{\sup}(W_{ij})^2]<\infty$ under Assumption \ref{a:M} (vi), by applying Theorem \ref{theorem:clt} under Assumptions \ref{a:sampling} and \ref{a:scalar}, we obtain $\sqrt{\underline{C}}\nabla_{\theta}\hat Q_{NM}\left(\theta^0\right)\overset{d}{\rightarrow}N\left(0,\Sigma\right)$, where $\Sigma=\Sigma_1+\Lambda \Sigma_2,$  
\begin{align*}
\Sigma_1=&\lambda_1E\left[\nabla_{\theta}q\left(W_{11},\theta^0\right) \nabla_{\theta}q\left(W_{12},\theta^0\right)^T\right]+\lambda_2E\left[\nabla_{\theta}q\left(W_{11},\theta^0\right) \nabla_{\theta}q\left(W_{21},\theta^0\right)^T\right],\\
 \Sigma_2=&E\left[\nabla_{\theta}q\left(W_{11},\theta^0\right) \nabla_{\theta}q\left(W_{11},\theta^0\right)^T\right].
\end{align*}
This implies that condition 3.1 (iii) is satisfied. 
Under Assumption \ref{a:M} (i), (ii) (c), similar lines of argument to those in the proof of Lemma \ref{theorem:consistency_M} yield $N\left(\mathcal{Q''},\left\|\cdot\right\|_{Q,2},\epsilon\right)<\infty$ for any finite discrete measure $Q$ and $\epsilon\in (0,1].$  
Note that $\mathcal{Q}''$ is a pointwise measurable class of functions by Lemma \ref{lemma:pointwise_measurable} under Assumptions \ref{a:M} (i), (ii) (c). With Assumptions \ref{a:sampling} and \ref{a:M} (iv), Lemma \ref{lemma:consistency} thus yields 
$$
\sup_{\theta \in \Theta}\left|\frac{1}{\hat L}\sum_{i=1}^{N}\sum_{j=1}^{M}Z_{ij}\frac{\partial^2 q\left(W_{ij}, \theta\right)}{\partial \theta_r\partial \theta_l}-E\left[\frac{\partial^2 q\left(W_{11}, \theta\right)}{\partial \theta_r\partial \theta_l}\right]\right|\overset{P}{\rightarrow}0
$$ 
for each $r,l\in \{1,...,k\}.$ 
Since there are only finite numbers of $r$ and $l$, we obtain $\sup_{\theta \in \Theta}\left\|\hat H\left(\theta\right)-H(\theta)\right\|\overset{P}{\rightarrow}0$, where $\hat H(\theta)=\nabla_{\theta\theta^T}\hat Q_{NM}\left(\theta\right).$ 
With Assumption \ref{a:M} (ii) (c), condition 3.1 (iv) is satisfied. Applying Theorem 3.1 in \citet{Newey1994}, we obtain $\sqrt{\underline C}\left(\hat \theta-\theta^0\right)\overset{d}{\rightarrow}N\left(0,H^{-1}\Sigma H^{-1}\right).$ 
\end{proof}

\subsection{Proof of Theorem \ref{theorem:consistency_M_variance}}\label{sec:theorem:consistency_M_variance}
\begin{proof}
First, we want to establish $\widetilde H\overset{P}{\rightarrow}H$ via
$
\left\|\widetilde H-H\right\|\leq \left\|\widetilde H-H\left(\hat \theta\right)\right\|+\left\|H\left(\hat \theta\right)-H\right\|,
$ 
where $H\left(\theta\right)=-E\left[\nabla_{\theta \theta^T}q\left(W_{11},\theta\right)\right].$ 
Since the conditions of Lemma \ref{theorem:consistency_M} are satisfied, we have 
$\hat \theta \overset{P}{\rightarrow}\theta^0.$ 
Under Assumption \ref{a:M:variance} (i), we thus obtain $\left\|H\left(\hat \theta\right)-H\right\|\overset{P}{\rightarrow}0$ by the continuous mapping theorem. 
Note that $\mathcal{Q}''$ is pointwise measurable and $N\left(\mathcal{Q''},\left\|\cdot\right\|_{Q,2},\epsilon\right)<\infty$ for any finite discrete measure $Q,$ $\epsilon\in (0,1]$ by the proof of Theorem \ref{theorem:asymptotic_normality_M}. 
Under Assumptions \ref{a:sampling} and \ref{a:M} (iv), by applying Lemma \ref{lemma:consistency}, we obtain 
$$
\sup_{\theta \in \Theta}\left|\frac{1}{\hat L}\sum_{i=1}^{N}\sum_{j=1}^{M}Z_{ij}\frac{\partial^2 q\left(W_{ij}, \theta\right)}{\partial \theta_r\partial \theta_l}-E\left[\frac{\partial^2 q\left(W_{11}, \theta\right)}{\partial \theta_r\partial \theta_l}\right]\right|\overset{P}{\rightarrow}0
$$ 
for each $r,l\in \{1,...,k\}.$ 
Since there are only finite numbers of $r$ and $l$, we get 
$$
\sup_{\theta\in \Theta}\left\|\frac{1}{\hat L}\sum_{i=1}^{N}\sum_{j=1}^{M}Z_{ij}\nabla_{\theta\theta^T}q\left(W_{ij},\theta\right)-E\left[\nabla_{\theta\theta^T}q\left(W_{11},\theta\right)\right]\right\|\overset{P}{\rightarrow}0.
$$ 
Thus, we obtain $\widetilde H\overset{P}{\rightarrow}H\left(\hat \theta\right).$ 
Combining the above result together yields $\widetilde H \overset{P}{\rightarrow} H.$ 
By the continuity of matrix inversion under Assumption \ref{a:M} (v), it follows that $\widetilde H ^{-1}\overset{P}{\rightarrow}H^{-1}.$ 

Next, we will establish $\widetilde \Sigma_2\overset{P}{\rightarrow}\Sigma_2.$
Define a new class $\mathcal{Q}_{\text{sub}}'=\bigl\{(w,z)\mapsto  z \partial q\left(w, \theta\right)/\partial \theta_r, \theta \in \Theta, \; r \in \left\{1,...,k\right\}\bigl\}.$ 
Under Assumption \ref{a:M} (i), (ii) (b), similar lines of argument to those in the proof of Lemma \ref{theorem:consistency_M} yield $\sup_{Q}N\left(\mathcal{Q}_{\text{sub}}',\left\|\cdot\right\|_{Q,2},\epsilon\left\|q_{\sup}\right\|_{Q,2}\right)<\infty,$ for any finite discrete measure $Q$ and $\epsilon\in (0,1].$ 
By Theorem 9.15 in \citet{kosorok2008introduction}, 
\begin{equation*}
     \begin{aligned}
      &\sup_{Q}N\left(\mathcal{Q}_{\text{sub}}'\mathcal{Q}_{\text{sub}}',\left\|\cdot\right\|_{Q,2},\sqrt{2}\epsilon\left\|q_{\sup}^2\right\|_{Q,2}\right)\\
      &\le \sup_{Q}N\left(\mathcal{Q}_{\text{sub}}',\left\|\cdot\right\|_{Q,2},\epsilon\left\|q_{\sup}\right\|_{Q,2}\right)\sup_{Q}N\left(\mathcal{Q}_{\text{sub}}',\left\|\cdot\right\|_{Q,2},\epsilon\left\|q_{\sup}\right\|_{Q,2}\right)<\infty
     \end{aligned}
\end{equation*}
holds for any finite discrete measure $Q$ and $\epsilon\in (0,1],$ where $\mathcal{Q}_{\text{sub}}'\mathcal{Q}_{\text{sub}}'$ is defined as the pointwise product. 
Note that $\mathcal{Q}_{\text{sub}}' \mathcal{Q}_{\text{sub}}'$ is a pointwise measurable class of functions since $\mathcal{Q}_{\text{sub}}'$ is a pointwise measurable class of functions by the same argument as in the proof of Theorem \ref{theorem:asymptotic_normality_M}. 
With Assumption \ref{a:sampling} and $E[\dot q_{\sup}(W_{ij})^2]<\infty$ under Assumption \ref{a:M} (vi), by applying Lemma \ref{lemma:d12},  
we obtain 
$$
E\left[\sup_{\theta \in \Theta}\left\|\frac{1}{L}\sum_{i=1}^{N}\sum_{j=1}^{M}\nabla_{\theta}q\left(W_{ij}, \theta\right)Z_{ij}\left(\nabla_{\theta}q\left(W_{ij}, \theta\right)Z_{ij}\right)^T-E\left[\nabla_{\theta}q\left(W_{11}, \theta\right)\nabla_{\theta}q\left(W_{11}, \theta\right)^T\right]\right\|\right]=o(1).
$$
Therefore, by Markov inequality, it follows that 
$$
\sup_{\theta \in \Theta}\left\|\frac{1}{L}\sum_{i=1}^{N}\sum_{j=1}^{M}\nabla_{\theta}q\left(W_{ij}, \theta\right)Z_{ij}\left(\nabla_{\theta}q\left(W_{ij}, \theta\right)Z_{ij}\right)^T-E\left[\nabla_{\theta}q\left(W_{11}, \theta\right)\nabla_{\theta}q\left(W_{11}, \theta\right)^T\right]\right\|\overset{P}{\rightarrow}0.
$$
In addition, we can write
$$
\frac{1}{L}\sum_{i=1}^{N}\sum_{j=1}^{M}Z_{ij}\nabla_{\theta}q\left(W_{ij}, \theta\right)\nabla_{\theta}q\left(W_{ij}, \theta\right)^T=\frac{1}{L}\sum_{i=1}^{N}\sum_{j=1}^{M}\nabla_{\theta}q\left(W_{ij}, \theta\right)Z_{ij}\left(\nabla_{\theta}q\left(W_{ij}, \theta\right)Z_{ij}\right)^T.
$$
Thus, Lemma \ref{lemma:inequality} yields $\widetilde \Sigma_2\overset{P}{\rightarrow}\Sigma_2(\hat \theta),$ where $\Sigma_2(\theta)=E\left[\nabla_{\theta}q\left(W_{11},\theta\right)\nabla_{\theta}q\left(W_{11},\theta\right)^T\right]$. 
Meanwhile, $\hat \theta\overset{P}{\rightarrow}\theta^0$ and we have $\Sigma_2(\hat \theta)\overset{P}{\rightarrow}\Sigma_2$ by Assumption \ref{a:M:variance} (iii). 
Therefore, we establish $\widetilde \Sigma_2\overset{P}{\rightarrow}\Sigma_2.$

Finally, we will establish $\widetilde \Sigma_1\overset{P}{\rightarrow}\Sigma_1$.
Note that $\mathcal{Q}_{\text{sub}}'$ is a pointwise measurable class of functions and $\sup_{Q}N\left(\mathcal{Q}_{\text{sub}}',\left\|\cdot\right\|_{Q,2},\epsilon\left\|q_{\sup}\right\|_{Q,2}\right)<\infty$, for any finite discrete measure $Q$, $\epsilon\in (0,1].$ 
With Assumptions \ref{a:sampling}, \ref{a:scalar} and $E[\dot q_{sup}(W_{ij})^2]<\infty$ under Assumption \ref{a:M} (vi), Lemma \ref{lemma:d11} yields\small 
\begin{equation*}
    \begin{aligned}
    &\lim_{\underline{C}\rightarrow \infty}E\left[\sup_{\theta \in \Theta}\left\|\frac{\underline{C}}{(NM)^2}\sum_{i=1}^{N}\sum_{1\leq j,j'\leq M}\nabla_{\theta}q\left(W_{ij},\theta\right)Z_{ij}\left(\nabla_{\theta}q\left(W_{ij'},\theta\right)Z_{ij'}\right)^T-E\left[\lambda_1 \nabla_{\theta}q\left(W_{11},\theta\right)Z_{11}\left(\nabla_{\theta}q\left(W_{12},\theta\right)Z_{12}\right)^T\right]\right\|\right]\\
   &=0.
\end{aligned}
\end{equation*}\normalsize
As we can write 
\begin{equation*}
    \begin{aligned}
    &\frac{\underline{C}}{(NM)^2}\sum_{i=1}^{N}\sum_{1\leq j,j'\leq M}\nabla_{\theta}q\left(W_{ij},\theta\right)Z_{ij}\left(\nabla_{\theta}q\left(W_{ij'},\theta\right)Z_{ij'}\right)^T\\
    &=\frac{L^2}{(NM)^2}\frac{\underline{C}}{L^2}\sum_{i=1}^{N}\sum_{1\leq j,j'\leq M}\nabla_{\theta}q\left(W_{ij},\theta\right)Z_{ij}\left(\nabla_{\theta}q\left(W_{ij'},\theta\right)Z_{ij'}\right)^T\\
    &=p^2\frac{\underline{C}}{L^2}\sum_{i=1}^{N}\sum_{1\leq j,j'\leq M}\nabla_{\theta}q\left(W_{ij},\theta\right)Z_{ij}\left(\nabla_{\theta}q\left(W_{ij'},\theta\right)Z_{ij'}\right)^T,
    \end{aligned}
\end{equation*}
and 
\begin{equation*}
    \begin{aligned}
    E\left[\lambda_1 \nabla_{\theta}q\left(W_{11},\theta\right)Z_{11}\left(\nabla_{\theta}q\left(W_{12},\theta\right)Z_{12}\right)^T\right]
    &=E\left[Z_{11}Z_{12}\right]E\left[\lambda_1 \nabla_{\theta}q\left(W_{11},\theta\right)\nabla_{\theta}q\left(W_{12},\theta\right)^T\right]\\
    &=p^2E\left[\lambda_1 \nabla_{\theta}q\left(W_{11},\theta\right)\nabla_{\theta}q\left(W_{12},\theta\right)^T\right].
    \end{aligned}
\end{equation*}
Therefore, by Markov inequality, it follows that
$$
\sup_{\theta\in\Theta}\left\|\frac{\underline{C}}{L^2}\sum_{i=1}^{N}\sum_{1\leq j,j'\leq M}\nabla_{\theta}q\left(W_{ij},\theta\right)Z_{ij}\left(\nabla_{\theta}q\left(W_{ij'},\theta\right)Z_{ij'}\right)^T-\lambda_1E\left[ \nabla_{\theta}q\left(W_{11},\theta\right)\nabla_{\theta}q\left(W_{12},\theta\right)^T\right]\right\|\overset{P}{\rightarrow}0
$$
as  $\underline{C}\rightarrow \infty.$
In addition, a symmetric argument also shows that
$$
\sup_{\theta \in \Theta}\left\|\frac{\underline{C}}{L^2}\sum_{1\leq i, i'\leq N}\sum_{j=1}^M\nabla_{\theta}q\left(W_{ij},\theta\right)Z_{ij}\left(\nabla_{\theta}q\left(W_{i'j},\theta\right)Z_{i'j}\right)^T-\lambda_2E\left[\nabla_{\theta}q\left(W_{11},\theta\right)\nabla_{\theta}q\left(W_{21},\theta\right)^T\right]\right\|\overset{P}{\rightarrow}0
$$ 
as $\underline{C}\rightarrow \infty.$
Also, we can write
$$
\frac{\underline{C}}{L^2}\sum_{i=1}^{N}\sum_{1\leq j,j'\leq M}Z_{ij}Z_{ij'}\nabla_{\theta}q\left(W_{ij},\theta\right)\nabla_{\theta}q\left(W_{ij'},\theta\right)^T=\frac{\underline{C}}{L^2}\sum_{i=1}^{N}\sum_{1\leq j,j'\leq M}\nabla_{\theta}q\left(W_{ij},\theta\right)Z_{ij}\left(\nabla_{\theta}q\left(W_{ij'},\theta\right)Z_{ij'}\right)^T
$$ 
and 
$$\frac{\underline{C}}{L^2}\sum_{1\leq i, i'\leq N}\sum_{j=1}^MZ_{ij}Z_{i'j}\nabla_{\theta}q\left(W_{ij},\theta\right)\nabla_{\theta}q\left(W_{i'j},\theta\right)^T=\frac{\underline{C}}{L^2}\sum_{1\leq i, i'\leq N}\sum_{j=1}^M\nabla_{\theta}q\left(W_{ij},\theta\right)Z_{ij}\left(\nabla_{\theta}q\left(W_{i'j},\theta\right)Z_{i'j}\right)^T.$$ 
Therefore, Lemma \ref{lemma:inequality} yields $\widetilde \Sigma_1\overset{P}{\rightarrow} \Sigma_1\left(\hat \theta\right),$ where 
$$ 
\Sigma_1(\theta)=\lambda_1E\left[\nabla_{\theta}q\left(W_{11},\theta\right)\nabla_{\theta}q\left(W_{12},\theta\right)^T\right]+ \lambda_2E\left[\nabla_{\theta}q\left(W_{11},\theta\right)\nabla_{\theta} q\left(W_{21},\theta\right)^T\right].
$$ 
Meanwhile, since $\hat \theta\overset{P}{\rightarrow} \theta^0,$ we get $\Sigma_1\left(\hat \theta \right) \overset{P}{\rightarrow}\Sigma_1$ by Assumption \ref{a:M:variance} (ii). We thus obtain $\widetilde \Sigma_1\overset{P}{\rightarrow}\Sigma_1$ and $\widetilde \Sigma\overset{P}{\rightarrow}\Sigma$ by continuous mapping theorem. 
Therefore, $\widetilde H ^{-1}\widetilde \Sigma \widetilde H ^{-1}$ is consistent for $H^{-1}\Sigma H^{-1}$ by the continuous mapping theorem. 
\end{proof}

\section{ Multiple Observations in a Cluster}\label{sec:multiple}
In many empirical applications, researchers face situations in which there are multiple observations in some cluster $(i,j)$, and the number of observations may vary across the clusters. In this section, we generalize the results from the main text to accommodate multiple observations per cluster. We shall follow the basic setting in \cite{davezies2018asymptotic}. Throughout this section, we call a pair $(i,j)$ of indices a cell. The number of observations in a cell is allowed to be random and can be correlated with the observations. This allows for a wide range of heterogeneous cluster structures. We will denote the number of observations in the $(i,j)$-th cell by $N_{ij}$, which is itself a random variable that takes a non-negative integer value. The observation that corresponds to the $\ell$-th unit, $1\le \ell\le N_{ij}$, in the $(i,j)$-th cell is a  $d$-dimensional  random vector denoted by $W_{\ell,ij}.$ 

With these notations, we consider the following sampling assumption.
\begin{assumption}[Sampling]\label{a:ext:sampling}
(i) The array $\left(N_{ij},(W_{\ell,ij})_{\ell \geq 1}\right)_{(i,j)\in \mathbb{N}^2}$ is separately exchangeable.
(ii) $\left(N_{ij},(W_{\ell,ij})_{\ell \geq 1}\right)_{\left(i,j\right)\in \mathbb{N}^2}$ is dissociated. 
(iii) $E[N_{11}]>0.$
\end{assumption}

\noindent 
This assumption is essentially identical to Assumption 1 in \citet*{davezies2018asymptotic}. Parts (i) and (ii) parallel Assumption \ref{a:sampling} (i) and (ii), respectively, in the main text except that the cell size is random and can differ across cells here. When $N_{ij}=1$ for all $i$ and $j$, the conditions reduce to Assumption \ref{a:sampling} in the main text. Also, part (ii) allows for a wide range of correlation structures between $N_{ij}$ and $(W_{\ell,ij})_{\ell \geq 1},$ and among $(W_{\ell,ij})_{\ell \geq 1}$ within $(i,j)$-th cell.  Part (iii) excludes the cells that are almost surely empty. 

\begin{assumption}[Function Class]\label{a:ext:moment}
(i) $E\left[\sum_{\ell=1}^{N_{ij}}f\left(W_{\ell,ij}\right)\right]=0.$ 
(ii) $\mathcal{F}$ admits an envelope $F$ satisfying $E\left[\sum_{\ell=1}^{N_{ij}}F(W_{\ell,ij})\right]<\infty$ with  $\sup_{Q} N(\mathcal{F},\left\|\cdot\right\|_{Q,2},\epsilon \left\|F\right\|_{Q,2})<\infty$ for all $\epsilon >0$, where $Q$ is any finite discrete measure. 
(iii) $\mathcal{F}$ is pointwise measurable.
\end{assumption}

\noindent 
This assumption generalizes Assumption \ref{a:moment}  by allowing for multiple observations within a cell as well as heterogeneous cell sizes.

\begin{lemma}[Uniform Weak Law of Large Numbers]\label{lemma:ext:consistency}
Suppose that Assumptions \ref{a:ext:sampling} and \ref{a:ext:moment} (ii)--(iii) hold.
Then we have 
$$ \sup_{f\in \mathcal{F}}\left|\frac{1}{\hat L}\sum_{i=1}^{N}\sum_{j=1}^{M}Z_{ij}\sum_{\ell=1}^{N_{ij}}f\left(W_{\ell,ij}\right)- E\left[\sum_{\ell=1}^{N_{11}}f(W_{\ell,11})\right]\right|\overset{P}{\rightarrow} 0.$$
\end{lemma}
\noindent
A proof of Lemma \ref{lemma:ext:consistency} is a straightforward modification of the proof of Lemma \ref{lemma:consistency} and is therefore omitted.  

\begin{theorem}[Central Limit Theorem]\label{theorem:ext:clt}
Suppose that Assumptions \ref{a:scalar}, \ref{a:ext:sampling} and \ref{a:ext:moment} (i), (iii) hold.
In addition, suppose that any finite function class $\mathcal{F}=\left\{f_1,...,f_k\right\}$ with a fixed $k$ admits an envelope $F$ satisfying $E\left[\left(\sum_{\ell=1}^{N_{ij}}F(W_{\ell,ij})\right)^2\right]<\infty$. Let $f=\left(f_1,...,f_k\right)^T$. 
Then,
\begin{equation*}
    \sqrt{\underline{C}}\frac{1}{\hat L}\sum_{i=1}^{N}\sum_{j=1}^{M}Z_{ij}\sum_{\ell=1}^{N_{ij}}f\left(W_{\ell,ij}\right)\overset{d}{\rightarrow}N\left(0,\Gamma\right),
\end{equation*}
where $\Gamma=\Gamma _A+\Lambda \Gamma_B$,
\begin{align*}
\Gamma_A=&\lambda_1 E\left[\left(\sum_{\ell=1}^{N_{11}}f\left(W_{\ell,11}\right)\right)\left(\sum_{\ell=1}^{N_{12}}f\left(W_{\ell,12}\right)\right)^T\right]+\lambda_2 E\left[\left(\sum_{\ell=1}^{N_{11}}f\left(W_{\ell,11}\right)\right)\left(\sum_{\ell=1}^{N_{21}}f\left(W_{\ell,21}\right)\right)^T\right],\\
\Gamma_B=&E\left[\left(\sum_{\ell=1}^{N_{11}}f\left(W_{\ell,11}\right)\right)\left(\sum_{\ell=1}^{N_{11}}f\left(W_{\ell,11}\right)\right)^T\right].
\end{align*} 

\end{theorem}
Again, a proof is a straightforward modification of that of Theorem \ref{theorem:clt}. Here we describe the necessary modification without repetitively showing the entire proof. Note that under the current setting, we have 
\begin{equation*}
\frac{1}{\hat L}\sum_{i=1}^{N}\sum_{j=1}^{M}Z_{ij}\sum_{\ell=1}^{N_{ij}}f\left(W_{\ell,ij}\right)=\frac{L}{\hat L}\left(A_{NM}+\sqrt{1-p_{NM}}B_{NM}\right),
\end{equation*}
where $A_{NM}$ and $B_{NM}$ are defined as
~\\
\begin{equation*}
    A_{NM}=\frac{1}{NM}\sum_{i=1}^{N}\sum_{j=1}^{M}\sum_{\ell=1}^{N_{ij}}f\left(W_{\ell,ij}\right) \quad \text{and} \quad  B_{NM}=\frac{1}{L}\sum_{i=1}^{N}\sum_{j=1}^{M}
\frac{Z_{ij}-p_{NM}}{\sqrt{1-p_{NM}}}\sum_{\ell=1}^{N_{ij}}f\left(W_{\ell,ij}\right),
\end{equation*}
respectively.
Then, using a similar argument to the one in the proof of Theorem \ref{theorem:clt}, the asymptotic normality for $A_{NM}$ can be established as 
\begin{align*}\label{eq:ext:a_normal}
&\sqrt{\underline{C}}A_{NM}\stackrel{d}{\rightarrow}\\
&N\left(0,\lambda_1 E\left[\left(\sum_{\ell=1}^{N_{11}}f\left(W_{\ell,11}\right)\right)\left(\sum_{\ell=1}^{N_{12}}f\left(W_{\ell,12}\right)\right)^T\right]+\lambda_2 E\left[\left(\sum_{\ell=1}^{N_{11}}f\left(W_{\ell,11}\right)\right)\left(\sum_{\ell=1}^{N_{21}}f\left(W_{\ell,21}\right)\right)^T\right]\right).
\end{align*}
Similarly, the variance-covariance matrices for $B_{NM}$ can be calculated as
\begin{equation*}\label{appen:eq:b_normal}
\begin{aligned}
\var\left(B_{NM}\right)=\frac{1}{L}E\left[\left(\sum_{\ell=1}^{N_{11}}f\left(W_{\ell,11}\right)\right)\left(\sum_{\ell=1}^{N_{11}}f\left(W_{\ell,11}\right)\right)^T\right].
\end{aligned}   
\end{equation*}

\noindent We thus obtain Theorem \ref{theorem:ext:clt} following from the arguments in the proof of Theorem \ref{theorem:clt}.
 
\subsection{Application to the Generalized Method of Moments}
We now generalize the results for GMM from Section \ref{sec:gmm} to allow for multiple observations per cell.
Under the current setting, we assume that the true parameter vector $\theta^0=(\theta_1^0,...,\theta_k^0)^T$  satisfies  $$E\left[\sum_{\ell=1}^{N_{ij}}g\left(W_{\ell,ij},\theta^0\right)\right]=0,$$
where $m\ge k$.    
The multiway algorithmic subsampling GMM estimator $\hat\theta$ can be subsequently defined as  
$$
\max_{\theta \in \Theta}\hat Q_{NM}\left(\theta\right),
$$ 
where  $\hat Q_{NM}\left(\theta\right)=-\hat g_{NM}\left( \theta\right)^T\hat V\hat g_{NM}\left(\theta\right)$, $\hat g_{NM}\left(\theta\right)=\hat L^{-1}\sum_{i=1}^{N}\sum_{j=1}^{M}Z_{ij}\sum_{\ell=1}^{N_{ij}}g\left(W_{\ell,ij},\theta\right)$.  
To establish the asymptotic properties of $\widehat\theta$, we impose the following conditions.

\begin{assumption}\label{ext:a:GMM}
~\\
(i) $V$ is positive semi-definite, and $VE\left[\sum_{\ell=1}^{N_{ij}}g(W_{\ell,ij},\theta)\right]=0$ iff $\theta=\theta^0.$

\noindent (ii) $\theta^0 \in \textrm{int} \left(\Theta\right)$, where $\Theta$ is a compact subset of $\mathbb{R}^k$.

\noindent (iii) (a) $\theta \mapsto g_{r}(w, \theta)$ is Lipschitz with a universal Lipschitz constant.  

(b) Each component of $\theta \mapsto \nabla_{\theta}g_r(w,\theta)$ is Lipschitz with a universal Lipschitz constant.

\noindent (iv) $E\left[\sup_{\theta\in \Theta}\left\|\sum_{\ell=1}^{N_{ij}}g\left(W_{\ell,ij},
\theta\right)\right\|\right]<\infty.$

\noindent (v) $G^TVG$ is nonsingular, where $G=E\left[\sum_{\ell=1}^{N_{ij}} \nabla_{\theta} g\left(W_{\ell,ij},\theta^0\right)\right].$

\noindent (vi) $E\left[\sup_{\theta \in\Theta}\left\| \sum_{\ell=1}^{N_{ij}} \nabla_\theta g\left(W_{\ell,ij},\theta\right)\right\|\right]<\infty.$

\noindent (vii) $g_{\sup}(\cdot)=\max_{r\in \{1,...,m\}}|g_r\left(\cdot,\theta\right)|$ satisfies $E\left[\left(\sum_{\ell=1}^{N_{ij}}g_{\sup}(W_{\ell,ij})\right)^2\right]<\infty$.
\end{assumption}

\begin{lemma}[Consistency of the Multiway Algorithmic Subsampling GMM Estimator]\label{ext:theorem:consistency_gmm}
If Assumptions \ref{a:ext:sampling} and \ref{ext:a:GMM} (i), (ii), (iii), (iv) hold, and that $\hat V\overset{P}{\rightarrow} V$,
then $\hat \theta \overset{P}{\rightarrow} \theta^0.$
\end{lemma}

\noindent A proof of Lemma \ref{ext:theorem:consistency_gmm} follows analogously from the proof of Lemma \ref{theorem:consistency_gmm} and an application of Lemma \ref{lemma:ext:consistency}. We omit the proof. 

\begin{theorem}[Asymptotic Normality of the Multiway Algorithmic Subsampling GMM Estimator]\label{ext: theorem:asymptotic_normality_gmm}
If Assumptions \ref{a:scalar}, \ref{a:ext:sampling} and \ref{ext:a:GMM} hold, and that $\hat V\overset{P}{\rightarrow} V$,
then $$\sqrt{\underline{C}}\left(\hat \theta-\theta^0\right)\overset{d}{\rightarrow} N\left(0,\left(G^TVG\right)^{-1}G^TV\Omega VG\left(G^TVG\right)^{-1}\right),$$
where $\, G=E\left[\sum_{\ell=1}^{N_{11}}\nabla_{\theta} g\left(W_{\ell,11},\theta^0\right)\right]$ and $\, \Omega=\Gamma _1+\Lambda \Gamma_2$, with \begin{equation*}
    \Gamma _1= \lambda_1E\left[\left(\sum_{\ell=1}^{N_{11}}g\left(W_{\ell,11},\theta^0\right)\right)\left(\sum_{\ell=1}^{N_{12}}g\left(W_{\ell,12},\theta^0\right)\right)^T\right]  + \lambda_2E\left[\left(\sum_{\ell=1}^{N_{11}}g\left(W_{\ell,11},\theta^0\right)\right)\left(\sum_{\ell=1}^{N_{21}}g\left(W_{\ell,21},\theta^0\right)\right)^T\right]
\end{equation*} and  $\Gamma_2=E\left[\left(\sum_{\ell=1}^{N_{11}}g\left(W_{\ell,11},\theta^0\right)\right)\left(\sum_{\ell=1}^{N_{11}}g\left(W_{\ell,11},\theta^0\right)\right)^T\right]$.
\end{theorem}

\noindent A proof of Theorem \ref{ext: theorem:asymptotic_normality_gmm} follows straightforwardly from the proof of Theorem \ref{theorem:asymptotic_normality_gmm} and an application of Theorem \ref{theorem:ext:clt} in place of Theorem \ref{theorem:clt}.  To avoid repetition, the proof is omitted.

\subsection{Application to the M-Estimation}
We now generalize the results for  M-Estimation from Section \ref{sec:m} to allow for multiple observations per cell.
Under this setting, the multiway algorithmic subsampling M-estimator $\hat \theta$ is defined as the solution to 
$$
\max_{\theta \in \Theta}-\frac{1}{\hat L}\sum_{i=1}^{N}\sum_{j=1}^{M}Z_{ij}\sum_{\ell=1}^{N_{ij}}q\left(W_{\ell,ij},\theta\right).
$$
The true parameter vector $\theta^0$ is assumed to be the unique solution of $\max_{\theta \in \Theta}-E\left[\sum_{\ell=1}^{N_{ij}}q\left(W_{\ell,ij},\theta\right)\right]$.
For each $\theta \in \Theta$, let $-\hat L^{-1}\sum_{i=1}^{N}\sum_{j=1}^{M}Z_{ij}\sum_{\ell=1}^{N_{ij}}q\left(W_{\ell,ij}, \theta\right)$ and $-E\left[\sum_{\ell=1}^{N_{ij}}q\left(W_{\ell,ij},\theta\right)\right]$ be denoted by $\hat Q_{NM}\left(\theta\right)$ and $Q_0\left(\theta\right)$, respectively.
To establish the asymptotic properties of $\hat \theta$, we restate Assumption \ref{a:M} as follows.

\begin{assumption}\label{ext:a:M}
~\\
(i) $\theta^0 \in \textrm{int} \left(\Theta\right)$ where $\Theta$ is a compact subset of $\mathbb{R}^k$. Also, $E\left[\sum_{\ell=1}^{N_{ij}}q(W_{\ell,ij}, \theta^0)\right]<E\left[\sum_{\ell=1}^{N_{ij}}q(W_{\ell,ij},\theta)\right]$ holds for all $\theta \in \Theta \backslash \{\theta_0\}$. 

\noindent (ii) (a) $\theta \mapsto q\left(w, \theta\right)$ is Lipschitz with a universal Lipschitz constant. 

(b) Each coordinate of $\theta \mapsto \nabla_{\theta}q(w,\theta)$ is Lipschitz with a universal Lipschitz constant. 

(c) Each coordinate of $\theta \mapsto \nabla_{\theta\theta^T}q(w,\theta)=\partial^2q\left(w,\theta\right)/\partial \theta \partial \theta^T$ is Lipschitz with a universal Lipschitz constant. 

\noindent (iii) $E[\sup_{\theta\in \Theta}\sum_{\ell=1}^{N_{ij}}q\left(W_{\ell,ij},\theta\right)]<\infty.$ 

\noindent (iv) $E\left[\sup_{\theta \in \Theta}\left\|\sum_{\ell=1}^{N_{ij}} \nabla_{\theta\theta^T} q\left(W_{\ell,ij},\theta\right)\right\|\right]<\infty.$ 

\noindent (v) $H=H\left(\theta^0\right)$ is nonsingular where $H(\theta)=-E\left[\sum_{\ell=1}^{N_{ij}} \nabla_{\theta\theta^T} q\left(W_{\ell,ij},\theta\right)\right].$

\noindent (vi)  $\dot q_{\sup}(\cdot)=\max_{r \in \left\{1,...,k\right\} }\left|\partial q(\cdot,\theta)/\partial \theta_r\right|$ satisfies $E\left[\left(\sum_{\ell=1}^{N_{ij}}\dot q_{\sup}(W_{\ell,ij})\right)^2\right]<\infty.$
\end{assumption}  
 
\begin{lemma}[Consistency of the Multiway Algorithmic Subsampling M-estimator]\label{ext:theorem:consistency_M}
If Assumptions \ref{a:ext:sampling} and \ref{ext:a:M} (i), (ii), (iii) hold, then, $\hat \theta \overset{P}{\rightarrow} \theta^0.$
\end{lemma}

\noindent
 A proof of Lemma \ref{ext:theorem:consistency_M}  follows analogously from the proof of Lemma  \ref{theorem:consistency_M} and an application of Lemma \ref{lemma:ext:consistency}. We omit the proof.

\begin{theorem}[Asymptotic Normality of the Multiway Algorithmic Subsampling M-estimator]\label{ext:theorem:asymptotic_normality_M}
If Assumptions \ref{a:scalar}, \ref{a:ext:sampling} and \ref{ext:a:M} hold, then $$\sqrt{\underline C}\left(\hat \theta-\theta^0\right)\overset{d}{\rightarrow}N\left(0,H^{-1}\Sigma H^{-1}\right),$$ where $H=-E\left[\sum_{\ell=1}^{N_{11}} \nabla_{\theta\theta^T} q\left(W_{\ell,11},\theta^0\right)\right]$, $\Sigma=\Sigma_1+\Lambda \Sigma_2,$ 
\begin{equation*}
\begin{aligned}
    \Sigma_1=&\lambda_1E\left[\left(\sum_{\ell=1}^{N_{11}}\nabla_{\theta}q\left(W_{\ell,11},\theta^0\right)\right) \left(\sum_{\ell=1}^{N_{12}}\nabla_{\theta}q\left(W_{\ell,12},\theta^0\right)\right)^T\right]\\
    &+ \lambda_2E\left[\left(\sum_{\ell=1}^{N_{11}} \nabla_{\theta} q\left(W_{\ell,11},\theta^0\right)\right) \left(\sum_{\ell=1}^{N_{21}}\nabla_{\theta}q\left(W_{\ell,21},\theta^0\right)\right)^T\right],
\end{aligned}
\end{equation*} and  $\, \Sigma_2=E\left[\left(\sum_{\ell=1}^{N_{11}} \nabla_{\theta} q\left(W_{\ell,11},\theta^0\right)\right)\left(\sum_{\ell=1}^{N_{11}}\nabla_{\theta}q\left(W_{\ell,11},\theta^0\right)\right)^T\right].$
\end{theorem}

\noindent 
A proof of Theorem \ref{ext:theorem:asymptotic_normality_M} follows straightforwardly from the proof of Theorem \ref{theorem:asymptotic_normality_M} and an application of Theorem \ref{theorem:ext:clt} in place of Theorem \ref{theorem:clt}. To avoid repetition, the proof is omitted.
\section {Alternative Subsampling Methods}\label{sec:multiple}
Although we have thus far focused on Bernoulli subsampling as the default subsampling method, other algorithmic subsampling schemes are also applicable.
In \citet*{LeeNg2020sketching}, two classes of algorithmic subsampling schemes are considered, namely, 
 random subsampling methods and random projection methods, with the Bernoulli subsampling considered throughout this paper belongs to the former category. It is possible to adapt other random subsampling methods under the multiway cluster sampling setting, such as uniform sampling with/without replacement. In fact, the robustness against possible degeneracy remains valid when either of these two alternative random subsampling schemes is substituted.  We are going to illustrate such adaptations in the rest of this section. 
On the other hand, as random projection-based methods produce rather different decompositions in the asymptotic terms,  their validity and asymptotic behaviors under the current setting remain unclear to us, and are therefore not discussed here. 

To implement uniform subsampling without replacement, the researcher sets $L$ randomly chosen $\{Z_{ij}:i=1,...,N, j=1,...,M\}$ to $1$ and the rest to $0$. 
To implement uniform subsampling with replacement, the researcher generates $\{Z_{ij}:i=1,...,N, j=1,...,M\}$ following a multinomial distribution with $L$ trials, mutually exclusive events $\{1,...,NM\}$ and equal event probabilities $1/(NM)$.
For both subsampling schemes, total number of subsampled units, $\sum_{i=1}^N\sum_{j=1}^MZ_{ij}=L$, is deterministic, while for
 the Bernoulli random sampling,  $\sum_{i=1}^N\sum_{j=1}^MZ_{ij}=\hat L$ is stochastic.  Despite of such a discrepancy, the uniform subsampling without replacement yields asymptotically the same result as Bernoulli subsampling, since Bernoulli subsampling can be considered as a uniform subsampling without replacement with a random sample size and also $\hat L/L\overset{P}{\rightarrow} 1$ due to Lemma \ref{lemma:inequality}. 
We state the following two propositions for the uniform subsampling with and without replacement.

\begin{proposition} \label{pro:withoutreplacement}
Consider the uniform subsampling without replacement.
If the conditions for Theorem \ref{theorem:clt} hold, then
\begin{equation}
\sqrt{\underline{C}}\frac{1}{ L}\sum_{i=1}^{N}\sum_{j=1}^{M}Z_{ij}f\left(W_{ij}\right) \overset{d}{\rightarrow}N\left(0,\Gamma_{UN}\right),
\end{equation}
\end{proposition}
\noindent  where $\Gamma_{UN}$ has the same form as $\Gamma$ in Theorem \ref{theorem:clt}.

\begin{proposition}\label{pro:withreplacement}
Consider the uniform subsampling with replacement.
If the conditions for Theorem \ref{theorem:clt} hold, then
\begin{equation}
\sqrt{\underline{C}}\frac{1}{ L}\sum_{i=1}^{N}\sum_{j=1}^{M}Z_{ij}f\left(W_{ij}\right) \overset{d}{\rightarrow}N\left(0,\Gamma_{UR}\right),
\end{equation}
\end{proposition}
\noindent  where $\Gamma_{UR}=\Gamma_A + \lim_{N,M\to\infty}(\underline C/(NMp))\Gamma_B$, where $\Gamma_A$ and $\Gamma_B$ are as defined in Theorem \ref{theorem:clt}.

Proofs of Propositions \ref{pro:withoutreplacement} and \ref{pro:withreplacement} follow from a straightforward adaptation of arguments in the proof of Theorem \ref{theorem:clt} with Lemma 2 and Theorem 1 of \citet*{Janson1984}. Here we describe the required modifications rather than reproducing the repetitive proofs. First note that under either of these subsampling schemes, one can proceed the following proof of Theorem \ref{theorem:clt} to obtain the decomposition of Equation (\ref{eq:CLT_decomp}) with the factor $L/\hat L$ replaced by $1$ .  In addition, the $B_{NM}$ term in the decomposition of Equation (\ref{eq:CLT_decomp}) has a different conditional (on observations) distribution  that depends on the subsampling scheme and thus a different conditional variance. The propositions then follow from calculating the alternative conditional variance of $B_{NM}$ and applying the law of total variance.

\end{appendix}

\bibliography{citation}
\end{document}